\documentclass{llncs}         
\usepackage[T1]{fontenc}
\usepackage[usenames,dvipsnames]{xcolor}
\usepackage{tikz}
\usetikzlibrary{shadows,decorations,shapes,decorations.text,decorations.markings,calc,arrows,spy}
\usepackage{mathtools}
\usepackage{dsfont}
\usepackage{amsthm}
\newtheorem{thm}{Theorem}
\newtheorem{lem}[thm]{Lemma}
\newtheorem{pro}[thm]{Proposition}

\usepackage[normalem]{ulem}
\usepackage{url}

\theoremstyle{remark}

\usepackage{hyperref}
\usepackage{nameref}
\makeatletter
\let\orgdescriptionlabel\descriptionlabel
\renewcommand*{\descriptionlabel}[1]{%
  \let\orglabel\label
  \let\label\@gobble
  \phantomsection
  \edef\@currentlabel{#1}%
  \let\label\orglabel
  \orgdescriptionlabel{#1}%
}
\makeatother

\usepackage{amsfonts}

\newcommand{\R}{\ensuremath{\mathbb{R}}}
\newcommand{\Fc}{\ensuremath{\mathcal{F}}}

\newcommand{\change}[1]{{#1}}  %changes by Anna and Fidel together June 19, 20

\newcommand{\rnote}[1]{{#1}} 

\newcommand{\pnote}[1]{{#1}} %changes by Anna June 26

\newcommand{\nnote}[1]{{#1}} %changes by Anna June 28 -- July 1

\newcommand{\mchange}[1]{{#1}} %changes by Fidel June 27

\newcommand{\Bchange}[1]{{#1}} %changes by Fidel June 27

\newcommand{\Achange}[1]{{#1}} %changes by Anna June 10 2014
\newcommand{\bnote}[1]{{#1}} % Fidel July 4

\newcommand{\rpnote}[1]{{#1}} % Fidel July 4.5

\newcommand{\rrnote}[1]{{#1}} % Anna July 4

\newcommand{\rfnote}[1]{{#1}} % Anna July 5

\newcommand{\wsnote}[1]{{#1}} % Fidel July 6

\newcommand{\longVerNote}[1]{{#1}} % Fidel Oct 15

\newcommand{\lAnote}[1]{{#1}} % Anna Oct 19

\newcommand{\lFBnote}[1]{{#1}} % Fidel Nov 18

\newcommand{\lFnote}[1]{{#1}} % Fidel Oct 24
   
\newcommand{\remove}[1]{{}}

\newcommand{\MoveToAppendix}[1]{{}}

\usepackage{appendix}

\begin{document}
%\linenumbers
\mainmatter

\title{Morphing Schnyder drawings of planar triangulations}
\titlerunning{Morphing Schnyder drawings}  % abbreviated title (for running head)
%                                     also used for the TOC unless
%                                     \toctitle is used

%\author{Fidel Barrera-Cruz\inst{1,2} \and Penny Haxell\inst{1,3} \and Anna
%  Lubiw\inst{1,3}} \authorrunning{Fidel Barrera-Cruz et al.}
\author{Fidel Barrera-Cruz \and Penny Haxell \and Anna
  Lubiw} \authorrunning{Fidel Barrera-Cruz et al.}

%%%% list of authors for the TOC (use if author list has to be modified)
\tocauthor{Fidel Barrera-Cruz, Penny Haxell, and Anna Lubiw}
\institute{University of Waterloo, Waterloo, Canada\\
{\tt \{fbarrera, pehaxell, alubiw\}@uwaterloo.ca}
%\and
%Partially supported by Conacyt
%\and
%Partially supported by NSERC
}

\maketitle              % typeset the title of the contribution
 
%\maketitle

\begin{abstract}
  We consider the problem of morphing between two planar drawings of
  the same triangulated graph, maintaining straight-line planarity.  A
  paper in SODA 2013 gave a morph that consists of $O(n^2)$ steps
  where each step is a linear morph that moves each of the $n$
  vertices in a straight line at uniform
  speed~\Bchange{\cite{Alamdari13}}.  However, their method imitates
  edge contractions so the grid size of the intermediate drawings is
  not bounded and the morphs are not good for visualization purposes.
  Using Schnyder embeddings, we are able to morph in $O(n^2)$ linear
  morphing steps and improve the grid size to $O(n)\times O(n)$ for a
  \Achange{significant}
 % large 
  class of drawings of triangulations, namely the class of weighted
  Schnyder drawings. The morphs are visually attractive.  Our method
  involves implementing the basic ``flip'' operations of Schnyder
  woods as linear morphs.  
\end{abstract}
  % We consider the problem of morphing between two planar drawings of
  % the same triangulated graph, maintaining straight-line planarity.  A
  % paper in SODA 2013 gave a morph that consists of a polynomial number
  % of steps where each step is a linear morph that moves each vertex in
  % a straight line at uniform speed.  However, the grid size of the
  % intermediate drawings was 
  % \pnote{not bounded and the morphs are not good for visualization purposes.} 
  % %bad.
  % Using Schnyder embeddings, we
  % improve the grid size to polynomial for a large class of
  % triangulations, namely the class of weighted Schnyder drawings.
  % \pnote{Our morphs are visually attractive.}

%%%%%%%%%%%%%%%%%%%%%%%%%
%%%%%%%%%%%%%%%%%%%%%%%%%
%\clearpage

\section{Introduction}\label{sec:intro}

Given a triangulation on $n$ vertices and two straight-line planar
drawings of it, $\Gamma$ and $\Gamma'$, that have the same unbounded
face, it is possible to morph from $\Gamma$ to $\Gamma'$ while
preserving straight-line planarity. This was proved by Cairns in
1944~\cite{Cairns44}.  
Cairns's proof is algorithmic but requires 
%An algorithm can be derived from Cairns's
%proof, but this algorithm requires 
exponentially many steps, where
each step is a \emph{linear morph} that moves every vertex in a
straight line at uniform speed.  Floater and
Gotsman~\cite{Floater1999} gave a polynomial time algorithm
% that solves this problem. Their solution uses 
using Tutte's graph drawing algorithm~\cite{Tutte1963}, but their
morph is not composed of linear morphs so the trajectories of the
vertices are more complicated, and there are no guarantees on how
close vertices and edges may become.  Recently, Alamdari et
al.~\cite{Alamdari13} gave a polynomial time algorithm based on
Cairns's approach that uses $ O(n^{2})$ linear morphs, \Bchange{and this has now
%  recently this has 
been improved to $O(n)$ %unidirectional morphs 
by  Angelini et al.~\cite{Angelini14}.}
% \change{but the grid size of the intermediate drawings has no known bound.}
 The main idea is to contract (or almost contract) edges.  With this
approach, perturbing vertices to prevent coincidence is already
challenging, and perturbing to keep them on a nice grid seems
impossible.
%However, using this approach there is no known bound on
%the grid size of the intermediate drawings. 

In this paper we propose a new approach to morphing based on Schnyder
drawings.  We give a planarity-preserving morph that is composed of
$O(n^2)$ linear morphs and for which the vertices of each of the
$O(n^2)$ intermediate drawings are on a $6n \times 6n$ grid.
% In this paper we give \pnote{a planarity-preserving morph that is composed of}
% %morphing algorithm 
% \change{
% %for a specific class of drawings of triangulations 
%   $O(n^2)$ linear morphs and for which each of the $O(n^{2})$ %the vertices of every
%   intermediate drawings are on 
%   \rnote{ a $6n \times 6n$ grid.} 
  % $ O(n)\times O(n)$ grid.  
Our algorithm works for \emph{weighted Schnyder drawings} which are
obtained from a Schnyder wood together with an assignment of positive
weights to the interior faces.
%The specific 
%drawings we refer to are those obtainable from a Schnyder
%wood together with an assignment of positive weights to the interior
%faces of the triangulation.
A Schnyder wood \lFnote{(see Section~\ref{sec:Schnyder})} is a special type of partition \nnote{(colouring)} and
orientation
% , or assignment of colours, 
of the edges of a planar triangulation into three rooted directed
trees.
Schnyder~\cite{Schnyder89,Schnyder90}
  used them to obtain straight-line planar drawings of
 triangulations
  in an $ O(n)\times O(n)$ grid.
  To do this he defined barycentric coordinates for each vertex in
  terms of the number of faces in certain regions of the Schnyder
  wood.
Dhandapani~\cite{Dhandapani10} noted that assigning \Achange{any} positive weights
to the faces still gives straight-line planar drawings.  We call these
\emph{weighted Schnyder drawings}---they are the drawings on
  which our morphing algorithm works.
  
Two weighted Schnyder drawings may differ in weights and in the
Schnyder wood.  We address these separately: we show that changing
weights corresponds to a single \pnote{planar} linear morph; altering
the Schnyder wood is more significant.
%the more significant aspect.

The set of Schnyder woods of a given planar triangulation forms a
distributive lattice~\cite{Brehm00,Felsner04,OssonaDeMendez94}
possibly of exponential size~\cite{Felsner07}.  The basic operation
for traversing this lattice is a ``flip'' that reverses a cyclically
oriented triangle and changes colours appropriately.  It is known that
the flip distance between two Schnyder woods in the lattice is
$ O(n^{2})$ (see Section~\ref{sec:Schnyder}).  Therefore, to morph
between two %weighted 
Schnyder drawings in $O(n^2)$ steps, it suffices
to show how a flip can be realized via a constant number of
planar linear morphs.  We show that flipping a facial triangle
corresponds to a single planar linear morph, 
\Achange{and that a flip of a separating triangle can be realized by three planar linear morphs.
}

\begin{figure}[ht]
  \centering
     \scalebox{.82}{\input{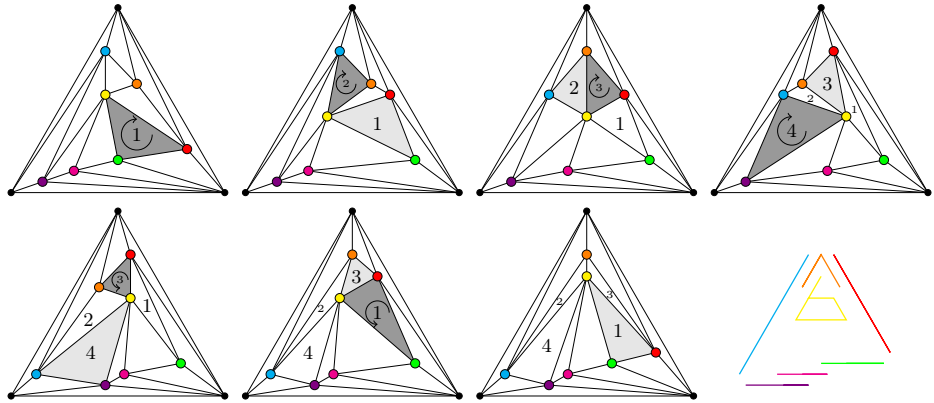} }    
  \caption{\wsnote{A sequence of triangle flips, counterclockwise along the top row and clockwise along the 
  bottom row.  In each drawing the triangle to be flipped is darkly shaded, and the one most recently flipped is lightly shaded.   
  The linear morph from each   drawing to the next one is planar. 
  Vertex trajectories are shown bottom right.
  }}
  \label{fig:anim_8_steps}
\end{figure}

\Achange{ There is hope that our method will give good visualizations
  for morphing.  See Figure~\ref{fig:anim_8_steps} \lAnote{and the larger example in 
Figure~\ref{fig:anim_16_steps}.  Further examples can be found in~\cite{BarreraCruz-web}.
\lFBnote{Note that the trajectories followed by vertices in a facial flip are
  always parallel to one of the three exterior edges.}
Our intermediate drawings
  lie on a $6n \times 6n$ grid where vertices are at least distance 1
  apart and face areas are at least $\frac{1}{2}$.
Comparing to visualization properties of previous methods,}    
the edge-contraction method of Alamdari et al.~\cite{Alamdari13} is
  \Achange{not good} %useless
  for visualization purposes---at the end of the recursion, the whole
  graph has contracted to a triangle. The method of Floater and
  Gotsman~\cite{Floater1999} gives good visualizations, based on
  experiments and heuristic improvements developed by Shurazhsky and
  Gotsman~\cite{Surazhsky-Gotsman}.  However, their method suffers the
  same drawbacks as Tutte's graph drawing method, namely that vertices
  and edges may come very close together.  }

\Achange{Not all straight-line planar triangulations are weighted
  Schnyder drawings, but we can recognize those that are in polynomial
  time~\lFnote{(see Section~\ref{sec:conclusions})}.  The problem of extending our result to all
  straight-line planar triangulations remains open.}
\Achange{There is partial progress in the first author's
  thesis~\cite{BarreraCruz14}.}
  %---an algorithm to morph from any  straight-line planar triangulation to a weighted Schnyder drawing in $O(n)$ steps, but not, however, on a nice grid.
%
%  \Bchange{However, it is possible to morph from any
%  straight-line planar triangulation to a weighted Schnyder drawing in
%$O(n)$ unidirectional morphing steps~\cite{BarreraCruz14}.}
%Should we define unidirectional?

This paper is structured as follows.  Section~\ref{sec:Schnyder}
contains the relevant background on Schnyder woods.
Section~\ref{sec:main} contains the precise statement of our main
result, and the general outline of the proof.  In
Section~\ref{sec:morph_weight_distribution} we show that changing
\pnote{face} weights corresponds to a linear morph.  Flips of facial
triangles are handled in Section~\ref{sec:face_morph} and flips of
separating triangles are handled in
Section~\ref{sec:separating_triangle_morph}. In
Section~\ref{sec:finding-weights} we explore which drawings are
weighted Schnyder drawings.

\begin{figure}[htbp]
  \centering   
\input{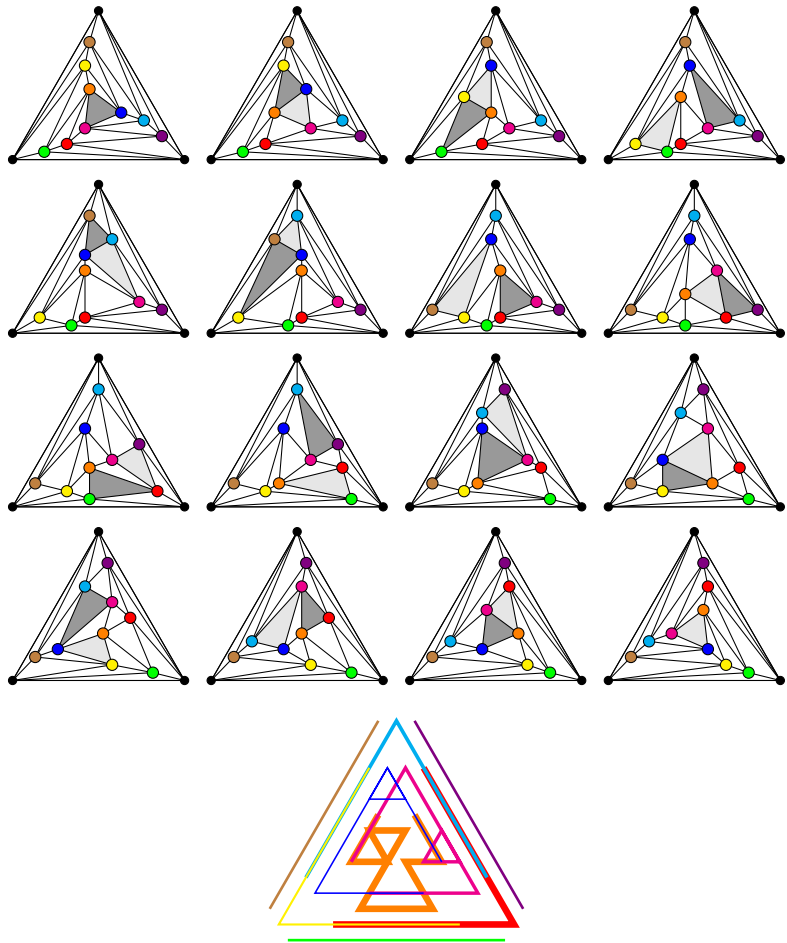}
  \caption{\wsnote{A sequence of triangle flips.  In each drawing the
      triangle to be flipped is darkly shaded, and the one most
      recently flipped is lightly shaded.
%   (counter clockwise) and flops
%    (clockwise). 
      The linear morph from each %one
      drawing to the next one is planar. At the bottom we illustrate
      the piecewise linear trajectories followed by vertices during
      the sequence of flips.}}
  \label{fig:anim_16_steps}
\end{figure}

\subsection{\pnote{Definitions and notation}}
\vspace{-.05in}
Consider two drawings $\Gamma$ and $\Gamma'$ of a planar triangulation
$T$. A \emph{morph} between $\Gamma$ and $\Gamma'$ is a continuous
family of drawings of $T$, $\{\Gamma^{t}\}_{t\in[0,1]}$, such that
$\Gamma^{0}=\Gamma$ and $\Gamma^{1}=\Gamma'$. We say a face $xyz$
\emph{collapses} during the morph $\{\Gamma^{t}\}_{t\in[0,1]}$ if
there is $t\in (0,1)$ such \Bchange{that} $x,y$ and $z$ are collinear in
$\Gamma^{t}$. We call a morph between $\Gamma$ and $\Gamma'$
\emph{planar} if $\Gamma^{t}$ is a planar drawing of $T$ for all
$t\in[0,1]$.  Note that a morph is planar
if and only if no face collapses during the morph. We call a morph
\emph{linear} if each vertex moves from its position in $\Gamma^{0}$
to its position in $\Gamma^{1}$ along a line segment and at constant
speed.  
\Achange{Note that each vertex may have a different speed.}
We denote such a linear morph by $\langle\Gamma^{0},\Gamma^{1}\rangle$.

%In this paper we only discuss straight-line drawings of planar
%triangulations, so \rfnote{(hmm) whenever we write drawing, we mean straight-line
%drawing.} That is, the drawings are uniquely
%determined by the positions of the vertices, since any edge is
%represented by the line segment joining its
%endpoints.

Throughout the paper we deal with a planar triangulation $T$
  with a distinguished exterior face 
%  whose vertices are
with vertices 
  $a_1, a_2,
  a_3$ in clockwise order.  
  The set of interior faces is denoted $\Fc(T)$.
%  We use
%  $\Fc(T)$ to denote the set of interior faces.}  
A 3-cycle $C$ whose removal disconnects $T$ is called a
  \emph{separating triangle}, and in this case we define $T|_{C}$ to
  be the triangulation formed by vertices inside $C$ together with $C$
  as the exterior face, and we define $T\setminus C$ to be the
  triangulation obtained from $T$ by deleting the vertices inside $C$.

%%%%%%%%%%%%%%%%%%%%%%%%%%%%%%%%%%%%%%%%%
%%%%%%%%%%%%%%%%%%%%%%%%%%%%%%%%%%%%%%%%%
\section{Schnyder woods and their properties}
\label{sec:Schnyder}

%\rfnote{This section contains necessary background on Schnyder woods and on the flip operation.}
%
%This section is devoted to introducing \change{background material
% relating} Schnyder woods. We show how Schnyder woods can be used to
% obtain planar drawings and also describe the flip operation in
% detail.
%
% \rfnote{Colors only for figures.} Throughout this paper we will use
% three colours red, green and blue; we will interchangeably refer to
% red, green and blue as 1, 2 and 3 respectively. When referring to
% the indices that represent colours, we will assume that the indices
% are reduced modulo 3.
%
A \emph{Schnyder wood}
of a planar triangulation $T$ 
\rfnote{with exterior vertices $a_1, a_2, a_3$} is an
assignment of directions and colours %red, green and blue 
\rfnote{1, 2, and 3} to the
interior edges of $T$ such that the following two conditions hold
\longVerNote{\Bchange{(see Figure~\ref{fig:sch_wood_def})}.}
\vspace{ -.05in}
\begin{description}
\item[\rfnote{(D1)}\label{enum:int_vx_prop}] \rfnote{Each interior
    vertex has three outgoing edges and they have colours 1, 2, 3 in
    clockwise order.  All incoming edges in colour $i$ appear between
    the two outgoing edges of colours $i-1$ and $i+1$ (index
    arithmetic modulo 3).
%whose colours are different from $i$. 
}
% Each interior vertex $v$ has outdegree 1 in colour $i$,
% $i=1,2,3$. At $v$, the outgoing edge in colour $i-1$, $e_{i-1}$,
% appears after the outgoing edge in colour $i+1$, $e_{i+1}$, in
% clockwise order. All incoming edges in colour $i$ appear in the
% clockwise sector between the edges $e_{i+1}$ and $e_{i-1}$.
\item[\rfnote{(D2)}\label{enum:ext_vx_prop}] At the exterior vertex
  $a_{i}$, all the interior edges are incoming and of colour $i$.
\end{description}
\begin{figure}[ht]
  \centering
  % The figure didn't like the scalebox
  %\input{figs/schn.tix}
%  \includegraphics[scale=0.7]{figs/schn2.pdf}
  \includegraphics[width=4in]{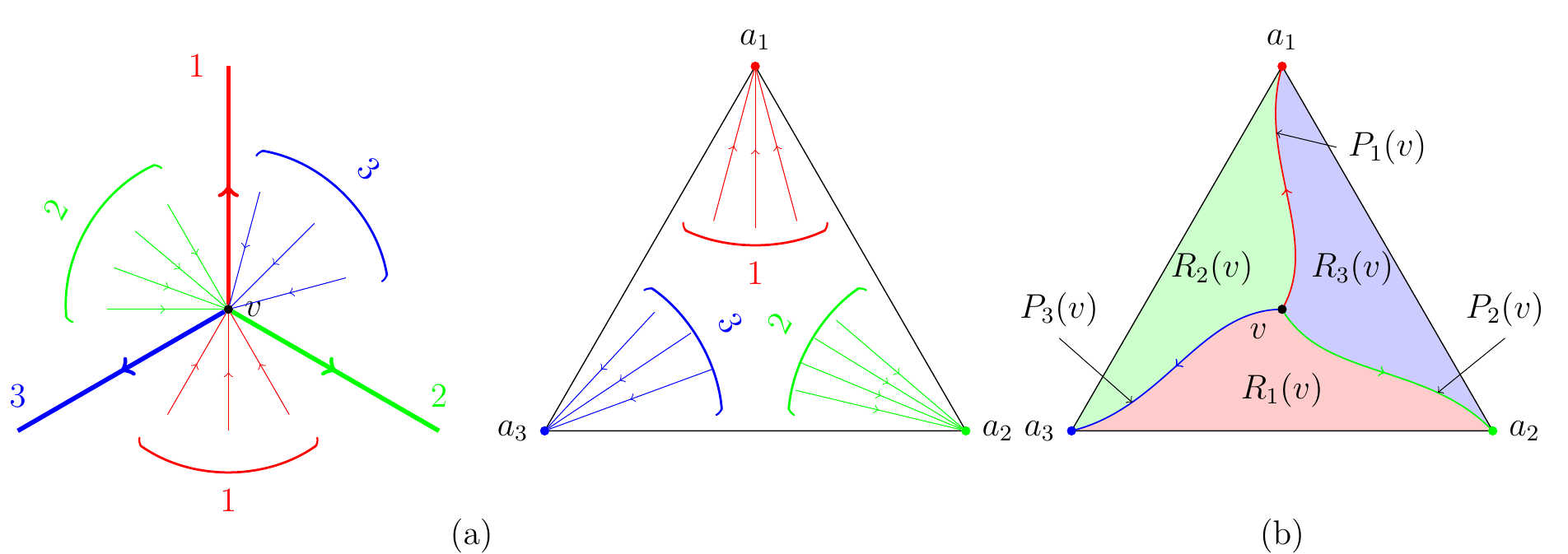}
  \caption{ (a) Conditions (D1) and
    (D2) for a Schnyder wood.
  (b) The paths and regions for vertex
      $v$. Throughout, we use red, green, blue for
      colours 1, 2, 3, respectively.}
  \label{fig:sch_wood_def} 
\end{figure}
% \rfnote{We don't use these names anywhere, so let's delete:} We call
% conditions~\ref{enum:int_vx_prop} and~\ref{enum:ext_vx_prop} the
% \emph{interior vertex property} and \emph{exterior vertex property}
% respectively.
   
\rfnote{The following basic concepts and properties are due to
Schnyder~\cite{Schnyder90}.
%Every planar triangulation admits a Schnyder wood, and 
%\rfnote{one can be found in linear time~\cite{Schnyder90}.}
%it can be
%obtained in $ O(n)$ time, where $n$ is the number of vertices of the
%planar triangulation~\cite{Schnyder90}.
For any Schnyder wood the edges of colour $i$ form a tree $T_i$ rooted
at $a_i$.  The path from internal vertex $v$ to $a_i$ in $T_i$ is
denoted $P_i(v)$.  } \remove{ A Schnyder wood defines a partition of
the interior edges of a planar triangulation into three edge disjoint
trees, one of each colour. Each of these trees can be thought of as
being rooted at an exterior vertex, that is, the $i$-tree is rooted at
$a_{i}$. Given a Schnyder wood $S$ of a planar triangulation we define
$P_{i}(v)$ to be the outgoing path from $v$ to $a_{i}$ in the $i$-tree
in $S$. There are two properties of the trees that we will find
useful.}
\begin{description}
\item [(P)\label{item:path_no_common_vertex}] If $T_{i}^{-}$ denotes
  the tree in colour $i$ with all arcs reversed, then $T_{i-1}^{-}\cup
  T_{i}\cup T_{i+1}^{-}$ contains no directed cycle. In particular,
  any two outgoing paths from a vertex $v$ have no vertex in common,
  except for $v$, i.e., $P_{i}(v)\cap P_{j}(v)=\{v\}$ for $i\neq
  j$.
%\item [(P2)] The outgoing path of colour $i$ leaving $v$ ends at $a_{i}$.
\end{description}
\rfnote{The \emph{descendants} of vertex $v$ in $T_i$, denoted
  $D_i(v)$, are the vertices that have paths to $v$ in $T_i$.  }
\remove{We say vertex $v$ is a \emph{descendant} of vertex $u$ in the
  $i$-tree if $u$ is a vertex of $P_{i}(v)$. For a vertex $u$, we
  define the set of descendants of $u$ in the $i$-tree
  as \[D_{i}(u)\coloneqq \{v\in V(T): v\text{ is a descendant of
  }u\text{ in the }i\text{-tree}\}.\] } \rfnote{For any interior
  vertex $v$ the three paths $P_i(v), i=1,2,3$ partition the
  triangulation into three regions $R_i(v), i=1,2,3$, where $R_i(v)$
  is bounded by $P_{i+1}(v),P_{i-1}(v)$ and $a_{i+1}a_{i-1}$.}%
%, see  Figure~\ref{fig:sch_wood_def}.}
%
\remove{
%By observing 
\rpnote{The paths outgoing from an interior vertex $v$
%, we can see that for each $v$ we obtain a 
partition of the set of interior faces of the planar triangulation
into three sets.} The $i$-th region, $R_{i}(v)$, is defined to be the
region bounded by $P_{i+1}(v),P_{i-1}(v)$ and $a_{i+1}a_{i-1}$. 
%We will also use $R_{i}(v)$ to denote the set of vertices incident to
%some face in the region. 
%\rpnote{The sets
%$R_{1}(v),R_{2}(v),R_{3}(v)$ do not partition on the set of
%vertices, as vertices of $P_{i}(v)$ belong to $R_{i-1}(v)$ and
%$R_{i+1}(v)$.}
%}
%In~\cite{Schnyder89,Schnyder90} the concept of \emph{barycentric
%  embedding} is introduced. We proceed to define this concept and
%then relate it to the regions defined by a Schnyder wood. 
%\rfnote{Skip barycentric embedding and go direct to the theorem.}
%\remove{
\rpnote{A \emph{barycentric embedding} of a graph $G$, a concept
  introduced in~\cite{Schnyder89,Schnyder90}, is an injective
  function} $f:V(G)\rightarrow \R^{3}$, $v\mapsto
(v_{1},v_{2},v_{3})$, so that the following two conditions hold.
\begin{itemize}
\item For each $v\in V(G)$, $v_{1}+v_{2}+v_{3}=2n-5$, and
\item for every $uv\in E(G)$ and $w\in V(G)\setminus\{u,v\}$ there is
  $k\in\{1,2,3\}$ so that $u_{k},v_{k}<w_{k}$.
\end{itemize}
Schnyder proved that if a graph admits a barycentric embedding then it
is a planar drawing.
%We now state a known result about barycentric embeddings.
%\begin{thm}[Schnyder~\cite{Schnyder89,Schnyder90}]\label{thm:schnyder_barycentric_embedding}
%  Each barycentric embedding of a graph defines a planar
%  drawing of a the graph in the plane $x+y+z=2n-5$.
%\end{thm}
In fact, barycentric embeddings for planar triangulations can be obtained from
Schnyder woods. This is stated in the following result.
}
\rfnote{ Schnyder proved that every triangulation $T$ has a Schnyder
  wood and that a planar drawing of $T$ can be obtained from
  coordinates that count faces inside regions: }

\begin{thm}[Schnyder~\cite{Schnyder89,Schnyder90}]\label{thm:schnyder_embedding}
  Let $T$ be a planar triangulation on $n$ vertices equipped with a
  Schnyder wood $S$. % be a Schnyder wood of $T$.
  Consider the map $f:V(T)\rightarrow\R^{3}$, where
  $f(a_{i})=(2n-5)e_{i}$, where $e_{i}$ denotes the $i$-th standard
  basis vector in $\R^{3}$, and for each interior vertex $v$,
  $f(v)=(v_{1},v_{2},v_{3})$, where $v_{i}$ denotes the number of
  faces contained inside region $R_{i}(v)$. Then $f$ \rfnote{defines a
     straight-line planar drawing.}
  %is a barycentric embedding.
\end{thm}

\rfnote{ Dhandapani~\cite{Dhandapani10} noted that the above result
  generalizes to weighted faces.  A \emph{weight distribution}
  $\mathbf{w}$ is a function that assigns a positive weight to each
  internal face such that the weights sum to $2n-5$.  For any weight
  distribution \lFnote{the \emph{$i$-th
    coordinate} $v_i$ of vertex $v$} is defined as:
\begin{equation}
v_{i}=\sum \{ \mathbf{w}(f) :  {f\in R_{i}(v)} \}.
\label{eq:weighted_drawing}
%v_{i}=\sum_{f\in R_{i}(v)}\mathbf{w}(f).\label{eq:weighted_drawing}
\end{equation}
\lFnote{Theorem~\ref{thm:schnyder_embedding} still holds if we use coordinates
as defined~\eqref{eq:weighted_drawing}.}
%to be the sum of the weights of faces inside $R_i(v)$.
We call the resulting straight-line planar drawing the
\emph{weighted Schnyder
  drawing} obtained from $\mathbf{w}$ and $S$.   
}
\remove{
As noted by Dhandapani in~\cite{Dhandapani10}, the previous result
holds more generally. 
We may consider a weight function
$\mathbf{w}:\Fc(T)\rightarrow (0,2n-5)$ that assigns to each face in
$\Fc(T)$ a positive weight such that
$\sum_{f\in\Fc(T)}\mathbf{w}(f)=2n-5$, we call such a function a
\emph{weight distribution}. 
Given a weight distribution we let
\begin{equation}
v_{i}=\sum_{f\in R_{i}(v)}\mathbf{w}(f).\label{eq:weighted_drawing}
\end{equation}

\rpnote{Dhandapani observed that mapping the exterior vertex $a_{i}$
  to $(2n-5)e_{i}$ and each interior vertex $v$ to
  $(v_{1},v_{2},v_{3})$ also defines a barycentric embedding.
% and therefore, by Theorem~\ref{thm:schnyder_barycentric_embedding}, a planar drawing.
We call such drawing $\Gamma$ of $T$ the \emph{weighted Schnyder
  drawing} obtained from $\mathbf{w}$ and $S$.  When the weight
distribution $\mathbf{w}$ is the constant $1$ function, the
corresponding weighted Schnyder drawing
% obtained from $\mathbf{w}$ 
coincides with
that from Theorem~\ref{thm:schnyder_embedding}.}}

%\rnote{The property stated here is only used in a proof
%from the Appendix.}
\longVerNote{ 
  \lFnote{Some useful properties} of weighted Schnyder drawings are
  stated next. 
%  As noted in the definition of barycentric embedding,
%  one of the conditions requires that for any edge $uv\in E(G)$ and
%  $x\in V(G)\setminus\{u,v\}$, there is a coordinate $k\in\{1,2,3\}$,
%  such that $u_{k},v_{k}<w_{k}$. This condition is fulfilled by the
%  following property of Schnyder woods.
\begin{description}
\item[(R1)\label{item:region_property}] Let $T$ be a planar
  triangulation and let $S$ be a Schnyder wood of $T$. Then, for any
  edge $uv\in E(T)$ and any $w\in V\setminus\{u,v\}$, there is
  $k\in\{1,2,3\}$ so that $u,v\in R_{k}(w).$ Consequently, in the
  corresponding weighted Schnyder drawing we have $u_{k},v_{k}<w_{k}$.
%\lAnote{move the following 2 items to join property R:}     
  \lFnote{\item[(R2)\label{item:region_property2}] The interiors of $R_{1}(x)$, $R_{2}(z)$ and $R_{3}(y)$ are
    pairwise disjoint.}

  \lFnote{\item[(R3)\label{item:region_property3}] %We have that 
  $D_{1}(x)\setminus\{x\}$ is contained in the
    interior of $R_{1}(x)$ and similarly for $y$ and $z$.
%    . Similar results hold for
%    $D_{2}(z)\setminus\{z\}$ and $ R_{2}(z)$, and $D_{3}(y)\setminus\{y\}$ and $
%    R_{3}(y)$ respectively. 
Consequently $D_{1}(x)$, $D_{2}(z)$ and
    $D_{3}(y)$ are pairwise disjoint.}
%The sets $D_{1}(x)$, $D_{2}(z)$ and $D_{3}(y)$ are pairwise
%    disjoint.
%    \rnote{How about strengthening this to: Except for vertex $x$, the set $D_1(x)$
%is strictly contained in $R_1(x)$, and similarly for $y$ and $z$.   Maybe we
%should also add that the interiors of the regions $R_1(x), R_2(y), R_3(z)$ are
%pairwise disjoint? 
    %In particular, the sets $D_{1}(x)$, $D_{2}(z)$ and $D_{3}(y)$ are pairwise
    %disjoint and no edge joins vertices is different sets other than $x,y,z$.
%}

\end{description}
}

%%%%%%%%%%%%%%%%%%%%%%%%%%%%%%%%%%%%%%%%%
%%%%%%%%%%%%%%%%%%%%%%%%%%%%%%%%%%%%%%%%%
%\subsection{Some properties of the set of Schnyder woods}
\lAnote{
\subsection{Triangle flips}
\label{subsec:flip_triangle}

In this section we}
%\Achange{We now}
describe results of Brehm~\cite{Brehm00},
  Ossona De Mendez~\cite{OssonaDeMendez94}, and
  Felsner~\cite{Felsner04} on the flip operation that can be used to
  convert any Schnyder wood to any other.
  Let $S$ be a Schnyder wood of planar triangulation $T$. A flip
  operates on a cyclically oriented triangle $C$ of $T$. 
%We begin by  proving two useful properties of Schnyder woods.
% We
%  \Bchange{use} the following properties of such a triangle (proofs in
%  the long version).
%%
\lAnote{Lemma~\ref{lem:cyclic_triangle} below gives a property of such cyclically oriented triangles.  Before that we need the following lemma about separating triangles.} 
\longVerNote{
%  \change{The following are two useful properties of Schnyder woods that
%    describe the structure of separating triangles and cyclically oriented triangles.}
 \begin{lem}\label{coro:mini_schnyder_wood}
   If $C$ is a separating triangle, then the restriction of $S$ to the
   interior edges of $T|_{C}$ is a Schnyder wood of $T|_{C}$.
 \end{lem}
  \begin{proof}%[Proof of property~\ref{coro:mini_schnyder_wood}]
    \rpnote{It is clear that the interior vertex property is satisfied
      in $T|_{C}$. We are just left to show the exterior vertex
      property. We may assume that $C$ is cyclically oriented in
      counterclockwise order, say $C=b_{1},b_{2},b_{3}$. The case
      where $C$ is clockwise will follow from a similar argument. We
      show that all interior edges of $T|_{C}$ at $b_{1}$ are incoming
      and have the same colour, the argument for $b_{2}$ and $b_{3}$
      is analogous. Suppose, by contradiction, that there is an
      interior edge of $T|_{C}$ in $S$ having colour $i$ that is
      outgoing from $b_{1}$. Assume $(b_{1},b_{2})$ has colour $j$ in $S$,
      clearly $j\neq i$.  Since $P_{i}(b_{1})$ ends at the exterior
      vertex $a_{i}$ it must
      leave $C$ through $b_{2}$ or $b_{3}$. Clearly $b_{2}\not\in
      P_{i}(b_{1})$, otherwise this would contradict
      property~\ref{item:path_no_common_vertex} for $P_{i}(b_{1})$ and
      $P_{j}(b_{1})$. Let $j'$ be the colour of $(b_{3},b_{1})$ and
      suppose $b_{3}\in P_{i}(b_{1})$. Clearly $j'\neq i$. We can now
      see that $T_{i}^{-}\cup T_{j'}^{-}$ contains a cycle, which
      contradicts property~\ref{item:path_no_common_vertex}.
      Therefore all interior edges of $T|_{C}$ at $b_{1}$ are
      incoming. Furthermore, these edges must be of the same colour or
      else, by the vertex property, there would be an outgoing edge
      towards the interior of $C$. To conclude the proof we show that
      the colours of interior edges of $T|_{C}$ at $b_{1}$, $b_{2}$
      and $b_{3}$ appear in the same cyclic order as the colours of
      the interior edges incident to the exterior vertices $a_{1}$ $a_{2}$ and $a_{3}$. Let
      $v$ be an interior vertex of $T|_{C}$ and suppose, without loss
      of generality, that $P_{i}(v)$ uses $b_{1}$. By
      property~\ref{item:path_no_common_vertex} it follows that
      $P_{i+1}(v)$ and $P_{i-1}(v)$ leave through $b_{3}$ and $b_{2}$
      respectively, otherwise, by the vertex property we would have
      that $P_{i+1}(v)$ and $P_{i-1}(v)$ would cross,
      contradicting~\ref{item:path_no_common_vertex}. This concludes
      the proof.  }
 \end{proof}
% , where $b_{i}$ is the
%   vertex of $C$ such that all the angles interior to $C$ at $b_{i}$
%   are labelled $i$ in the corresponding Schnyder labelling, here the
%   existence of each $b_{i}$ follows from 
%   Theorem~\ref{thm:schnyder_cycle}. Therefore all the interior angles
%   in $T|_{C}$ at $b_{i}$ are labelled $i$. It follows that all the
%   interior edges of $T|_{C}$ incident to $b_{i}$ are directed towards
%   to $b_{i}$ and have colour $i$ in $S$, as desired.

\begin{lem}\label{lem:cyclic_triangle}
  Let $T$ be a planar triangulation and let $S$ be a Schnyder wood of
  $T$. If $S$ has a triangle $C$ at which the edges are oriented
  cyclically, then $C$ has an edge of each colour in $S$. Furthermore,
  if $C$ is oriented counterclockwise then the edges along the
  cycle have colours $i$, $i-1$ and $i+1$ respectively.
\end{lem}
 \begin{proof}%[Proof of property~\ref{lem:cyclic_triangle}]
   \rpnote{Let $C=b_{1}b_{2}b_{3}$ and suppose the arcs forming $C$ are
     $\alpha_{1}=(b_{1},b_{2})$, $\alpha_{2}=(b_{2},b_{3})$ and
     $\alpha_{3}=(b_{3},b_{1})$. By
     property~\ref{item:path_no_common_vertex} not all arcs have the
     same colour. If two of the arcs share colour $i$, say $\alpha_{1}$
     and $\alpha_{2}$, and $\alpha_{3}$ has colour $j$, then
     $T_{i}^{-}\cup T_{j}^{-}$ contains a cycle, which contradicts
     property~\ref{item:path_no_common_vertex}. Therefore the colours
     of $\alpha_{1}$, $\alpha_{2}$ and $\alpha_{3}$ are pairwise
     distinct. Now, let us assume that $C$ is oriented
     counterclockwise. Suppose without loss of generality that
     $\alpha_{1}$ has colour $i$. It suffices to show that $\alpha_{2}$
     has colour $i-1$ and by rotational symmetry the result will
     follow. Suppose, by contradiction that $\alpha_{2}$ has colour
     $i+1$. If $C$ is a facial triangle, then this contradicts the vertex
     property at vertex $b_{2}$, since an incoming arc in colour $i$
     cannot be followed in clockwise order by an outgoing arc in colour
     $i+1$. Now, if $C$ is a separating triangle, then there are two
     interior edges of $T|_{C}$ that are outgoing from $b_{2}$, and this
     contradicts Lemma~\ref{coro:mini_schnyder_wood}. Therefore
     $\alpha_{2}$ must have colour $i-1$. The result now follows.}
 \end{proof}
}

%%
% We summarize Lemmas~\ref{lem:lem:cyclic_triangle}
% and~\ref{coro:coro:mini_schnyder_wood} 
% \begin{description}
% \item [(S1)\label{lem:cyclic_triangle}] The triangle $C$ has an edge
%   of each colour in $S$. Furthermore, if $C$ is oriented
%   counterclockwise then the edges along $C$ have colours $i$, $i-1$, $i+1$.
% \item [(S2)\label{coro:mini_schnyder_wood}] If $C$ is a separating
%   triangle, then the restriction of $S$ to
%   the interior edges of $T|_{C}$ is a Schnyder wood of $T|_{C}$.
% \end{description}

\lAnote{Let $C$ be a cyclically oriented triangle in $T$.  By Lemma~\ref{lem:cyclic_triangle}, we may assume that} $C=xyz$ oriented counterclockwise with edges $xy, yz, zx$ of
colour $1,3,2$ respectively.
%Let $C=xyz$ be oriented counterclockwise with edges $xy, yz, zx$ of colour $1,3,2$ respectively.  
A \emph{clockwise flip} of $C$ alters
the colours and orientations of $S$ as follows:
\begin{enumerate}
\item Edges on the cycle are reversed and colours change from $i$ to
  $i-1$. \Bchange{See triangle $xyz$ in
    Figure~\ref{fig:face_flip_coords}.}
\item Any interior edge of $T|_{C}$ remains with the same orientation
  and changes colour from $i$ to $i+1$.  \Bchange{See edges incident
    to $b$ in Figure~\ref{fig:b_coords}.}
\end{enumerate}
Other edges are unchanged.  The reverse operation is a
\emph{counterclockwise flip}, which Brehm calls a \emph{flop}.
Brehm~\cite[p. 44]{Brehm00} proves that a flip yields another Schnyder
wood.  Consider the graph with a vertex for each Schnyder wood of $T$
and a directed edge $(S,S')$ when $S'$ can be obtained from $S$ by a
clockwise flip.  This graph forms a distributive
lattice~\cite{Brehm00,Felsner04,OssonaDeMendez94}.  Ignoring edge
directions, the distance between two Schnyder woods in this graph is
called their \emph{flip distance}.
\longVerNote{
We now work towards proving that such flip distance is in fact
$O(n^{2}),$ where $n$ is the number of vertices of the planar
triangulation. 

  A triangle in a Schnyder wood $S$ is called \emph{flippable} if it
  is cyclically oriented counterclockwise. \emph{Floppable} triangles
  are triangles whose edges are cyclically oriented clockwise. Let $C$
  be a flippable triangle in $S$ and denote by $S^{C}$ the Schnyder
  wood obtained from flipping $C$ in $S$. We say
  $C_{1},C_{2},\ldots,C_{k}$ is a flip sequence if $C_{1}$ is
  flippable in $S_{1}\coloneqq S$, and $C_{i}$ is flippable in $S_{i}$
  with $S_{i+1}\coloneqq S_{i}^{C_{i}}$ for $1\leq i\leq k-1$. Note
  that if $C_{1},\ldots,C_{k}$ defines a flip sequence, then
  $C_{k},\ldots,C_{1}$ defines a flop sequence. In a 4-connected
  planar triangulation $T$ the flippable triangles are precisely the
  cyclically oriented faces of $T$.

  The following result of Brehm will allow us to derive an upper
  bound for the length of a maximal flip sequence in any planar
  triangulation.

  \begin{pro}[{Brehm~\cite[p. 15]{Brehm00}}]\label{pro:adjacent_faces_bound}
    Let $T$ be a 4-connected planar triangulation, let $S$ be a
    Schnyder wood of $T$ and let $f$ be an interior face of $T$.  If
    there is a flip sequence that contains $f$ at least twice, then
    between any two flips of $f$, all the faces adjacent to $f$ must
    be flipped.
  \end{pro}
  
%  \begin{thm}[{Brehm~\cite[p. 15]{Brehm00}}]
%    The length of any flip sequence in a 4-connected planar
%    triangulation is bounded.
%  \end{thm}
  Let us observe that any 
  %From the previous proposition we can conclude that if we consider any
  maximal flip sequence starting at an arbitrary Schnyder wood ends at
  the unique Schnyder wood containing no flippable face. In fact we
  can derive the following bound, as suggested
  in~\cite[p. 15]{Brehm00}.

  \begin{thm}\label{thm:brehm_4_connected_sequence}
    Let $T$ be a 4-connected planar triangulation on $n$ vertices with
    exterior face $f^{*}$. The length of any flip sequence $F$ is
    bounded by the sum of the distances in the dual of $T$ to $f^{*}$,
    that is, \[|F|\leq \sum_{f\in\Fc(T)}\text{d}(f,f^{*})= O(n^{2}).\]
    Furthermore, any maximal flip sequence terminates at $\mathcal{L}$,
    the Schnyder wood containing no flippable triangles.
  \end{thm}
  \begin{proof}
    Observe that any face adjacent to the exterior face cannot be
    flipped. By Proposition~\ref{pro:adjacent_faces_bound} it follows
    that the number of times a face $f$ can be flipped is bounded by
    the distance from $f$ to the exterior face $f^{*}$. Therefore the
    length of any flip sequence is $ O(n^{2})$, as claimed. Clearly
    any maximal flip sequence terminates at the Schnyder wood
    containing no flippable face.%, or otherwise it would contradict its maximality.
  \end{proof}

  A $4$-connected block of a graph $G$ is a maximal $4$-connected
  induced subgraph of $G$. If $G_{1},\ldots,G_{k}$ denote the
  $4$-connected blocks of a graph $G$ then we say that they define a
  \emph{decomposition of $G$ into $4$-connected blocks.}  One more
  result that will be useful is the following.

  \begin{thm}[{Brehm~\cite[p. 37]{Brehm00}}]\label{thm:brehm_flip_sequence_decomposition}
    Consider a planar triangulation $T$. Let $T_{1},\ldots,T_{k}$ be a
    decomposition of $T$ into $4$-connected blocks. Then any flip
    sequence in $S$ can be obtained by concatenating flip sequences
    $S_{1},S_{2},\ldots,S_{k}$, where each $S_{i}$ is flip sequence in
    $T_{i}$.
  \end{thm}

  We now prove that the flip distance between any two Schnyder woods
  of a planar triangulation on $n$ vertices is $O(n^{2})$.
%    Lemma~\ref{lem:lattice_not_4_connected_sequence}.}    
  \Achange{We have attributed the following result to
    Brehm~\cite{Brehm00}, but he does not state it as a single result.
  % Clearer descriptions of Brehm's results were given by Miracle et
  % al.~\cite{Miracle} and Eppstein et al.~\cite{Eppstein}.  }
  %   \rfnote{We note that Lemma~\ref{lem:lattice_not_4_connected_sequence}
  % is not stated as a single result in Brehm~\cite{Brehm00}. 
  It helps to read Miracle et al.~\cite{Miracle} and Eppstein et
  al.~\cite{Eppstein}.}

\begin{lem}% [Brehm (see the long version)]
  \label{lem:lattice_not_4_connected_sequence}
  In a planar triangulation on $n$ vertices the flip distance 
  % Let $T$ be a planar triangulation on $n$ vertices. The flip distance 
  between any two Schnyder woods 
  % of $T$ 
  is $O(n^{2})$, \rnote{and a flip sequence of that length can be found in linear time per flip.} 
\end{lem}
\begin{proof}
  \wsnote{We begin by showing that the length of a maximal flip
    sequence is $O(n^{2})$.} Let $T$ be a planar triangulation on
  $n$ vertices. Consider a decomposition of $T$ into $4$-connected
  blocks $T_{1},\ldots,T_{k}$. Denote by $n_{i}$ the number of
  interior vertices of $T_i$, $1\leq i\leq k$. Therefore we have
  $3+\sum_{i=1}^{k}n_{i}=n$. Observe that the length of a maximal
    flip sequence in the lattice of Schnyder woods of $T_{i}$ is
    $ O(n_{i}^{2})$ from
    Theorem~\ref{thm:brehm_4_connected_sequence}. By
    Theorem~\ref{thm:brehm_flip_sequence_decomposition} the length of
    a maximal flip sequence in the lattice of Schnyder woods of $T$ is
    $ O(\sum_{i=1}^{k}n_{i}^{2})$. We now have the following standard
    claim.
  
    \begin{claim}
      The value of $\sum_{i=1}^{k}n_{i}^{2}$ subject to
      $\sum_{i=1}^{k}n_{i}=n-3$ and $n_{i}\geq 0$, is maximized when
      exactly one of the $n_{i}$ is equal to $n-3$ and all others are
      zero.
    \end{claim}
    \begin{proof}[Proof of claim]\renewcommand{\qedsymbol}{}
      Suppose at least two terms are non zero say $0<n_{1}\leq
      n_{2}$. Let us show how to increase the value of the sum.
      \[
      \begin{split}
        (n_{1}-1)^{2}+(n_{2}+1)^{2}+\sum_{i=3}^{k}n_{i}^{2}&=2(n_{2}-n_{1}+1)+\sum_{i=1}^{k}n_{i}^{2}\\
        &> \sum_{i=1}^{k}n_{i}^{2}.
      \end{split}
      \]
      So the claim holds.
    \end{proof}
    
    A consequence of the claim is that
    $ O(\sum_{i=1}^{k}n_{i}^{2})= O(n^{2})$ and therefore the length
    of a maximal flip sequence in the lattice of Schnyder woods of $T$
    is $ O(n^{2})$. 

    \wsnote{Now, since any maximal flip sequence terminates at the
      Schnyder wood $\mathcal{L}$ that contains no flippable triangle,
      this yields a walk through $\mathcal{L}$ of length $O(n^{2})$
      between any two Schnyder woods $S$ and $S'$.

\lAnote{Finally we prove that the flip sequence can be found in linear time per flip.
%      Finally, for Schnyder woods $S$ and $S'$ that differ by a flip,
     If Schnyder woods $S$ and $S'$ differ by a flip,}
      we can obtain $S'$ from $S$ by reversing $O(1)$ arcs and by
      updating the colour of $O(n)$ arcs. The result now follows.}
  \end{proof}
}

\remove{
\rfnote{We note that Lemma~\ref{lem:lattice_not_4_connected_sequence}
  is not stated as a single result in Brehm~\cite{Brehm00}.  It helps
  to read Miracle et al.~\cite{Miracle} and Eppstein et
  al.~\cite{Eppstein}.}%  More details are in the long version.}
  }

\remove{ \change{In general, there are planar triangulations which
    admit exponentially many Schnyder woods as shown
    in~\cite{Felsner07}}. It has been proved that the set of Schnyder
  woods is a distributive
  lattice~\cite{Brehm00,Felsner04,OssonaDeMendez94}. In this
  distributive lattice, the basic operation to traverse it is by
  reversing cyclically oriented triangles and changing colours
  appropriately, we call this operation a \emph{flip} and \rpnote{it
    is described} at the end of this subsection. We provide a bound on
  the distance between any two Schnyder woods in the lattice.

  \remove{ We begin by introducing \emph{Schnyder labellings} and some
    of their properties in order to use them in the context of
    Schnyder woods.  Let $T$ be a planar triangulation. A
    \emph{Schnyder labelling} of $T$ is a labelling of the interior
    angles of $T$ with labels 1, 2 and 3 such that the following two
    conditions hold.
\begin{itemize}
\item[(Q1)] At any interior vertex $v$, the angles at $v$ form 3 non
  empty intervals. The labels in these intervals are 1, 2 and 3 in
  clockwise order.
\item[(Q2)]\label{item:sch_label_2} The angles of an interior face are
  labelled 1, 2 and 3 in clockwise order.
\end{itemize}

These labellings were first introduced by Schnyder
in~\cite{Schnyder89}, and there is a bijection between the set of
Schnyder labellings and Schnyder woods~\cite{Felsner08}, see
Figure~\ref{fig:sch_lab_sch_wood}. We now state one useful property of
Schnyder labellings.
\begin{figure}[!ht]
  \centering
  \input{figs/labeling_wood_bijection.tex}
  \caption{Bijection between Schnyder labellings and Schnyder woods.}
  \label{fig:sch_lab_sch_wood}
\end{figure}
}

\rpnote{Let us begin by stating two properties regarding cyclically
  oriented triangles. Let $T$ be a planar triangulation equipped with
  a Schnyder wood $S$ and let $C$ be a cyclically oriented triangle.
\begin{description}
\item [(S1)\label{lem:cyclic_triangle}] The triangle $C$ has an edge
  of each colour in $S$. Furthermore, if $C$ is oriented
  counterclockwise then the edges along $C$ have colours $i$, $i-1$,
  $i+1$.
\item [(S2)\label{coro:mini_schnyder_wood}] If $C$ is a separating
  triangle, then the restriction of $S$ to
  the interior edges of $T|_{C}$ is a Schnyder wood of $T|_{C}$.
\end{description}
%The proofs of these properties can be found in the long version.
\rfnote{Why are we proving these?  We should just refer to Schnyder or
  Felsner.}  }
%\rnote{Schnyder labellings were only used in the following two lemmas.}

% \begin{thm}[Schnyder~\cite{Schnyder89}]\label{thm:schnyder_cycle}
%   Let $T$ be a planar triangulation equipped with a Schnyder
%   labelling. If $C$ is a cycle in $T$, then for each $i\in\{1,2,3\}$
%   there is a vertex $v$ of $C$ such that all angles at $v$ that are
%   interior to $C$ are labelled $i$.
% \end{thm}

% A consequence of the previous result is the following.
%Given a Schnyder wood $S$ of a planar triangulation $T$, we can obtain
%another Schnyder wood $S'$ of $T$ different from $S$, provided that
%$S$ contains a triangle $C=xyz$ which is cyclically oriented. 
\rpnote{Using notation as above for $S$ and $C$ we will describe how
  to obtain another Schnyder wood $S'$.}  Assume that the
edges of $C$ are directed counterclockwise. 
%\rpnote{(a similar procedure works when $C$ is oriented clockwise). 
%We may assume without loss of generality and by
\rpnote{By Lemma~\ref{lem:cyclic_triangle} we may assume that}
the directed edges in $S$ are $(x,y), (y,z)$ and $(z,x)$, and that
they are coloured $1$, $3$ and $2$ respectively. We obtain $S'$ from
$S$ in the following way:
\begin{enumerate}
\item Every edge having an endpoint outside $xyz$ remains unchanged.
\item Edges on the cycle are reversed and colours change from $i$ to
  $i-1$. 
\item Any interior edge of $T|_{C}$ remains with the same orientation
  and changes colour from $i$ to $i+1$.
\end{enumerate}
We call such procedure a \emph{flip} of the triangle $xyz$. We define
a \emph{flop} to be the analogous procedure applied to a clockwise
directed triangle. The fact that performing a flip or a flop yields
another Schnyder wood is a particular case of a result proven by
Brehm~\cite[p. 44]{Brehm00}.
 
\remove{
\begin{thm}[Brehm~\cite{Brehm00}]\label{thm:brem_flip}
Let $T$ be a planar triangulation, let $S$ be a Schnyder wood of $T$
such that $S$ contains a triangle which is oriented
counterclockwise and let $S'$ be as described above. Then $S'$ is a
Schnyder wood of $T$.   
\end{thm}}

The previous \rpnote{construction defines the basic operation to
  traverse the lattice of Schnyder woods.} 
%provides a way to go from a Schnyder wood to
%another in the lattice. In fact, this is the basic step to traverse
%such structure.
 The next result gives a bound on the distance between
any two Schnyder woods in the Schnyder wood lattice. 
%Details about this result are provided in long version.%~\ref{app:sub:A}.
\begin{lem}\label{lem:lattice_not_4_connected_sequence}
  Let $T$ be a planar triangulation on $n$ vertices. The distance in
  the Schnyder wood lattice between any two Schnyder woods is
  $ O(n^{2})$.
\end{lem}

\rfnote{We note that Lemma~\ref{lem:lattice_not_4_connected_sequence}
  is not stated as a single result in Brehm~\cite{Brehm00}.  Clearer
  descriptions of Brehm's results were given by Miracle et al.~\cite{}
  and Eppstein et al.~\cite{}.}

}

%%%%%%%%%%%%%%%%%%%%%%%%%%%%%%%%%%%%%%%%%
%%%%%%%%%%%%%%%%%%%%%%%%%%%%%%%%%%%%%%%%%
%\section{Morphing Weighted Schnyder Drawings}
\section{\pnote{Main result}}
\label{sec:main}

%\pnote{Put main theorem in most general form (arbitrary weights).  Statement of theorem and THEN ingredients.}

%\mchange{In this section we introduce concepts related to morphs and we state our main result.}
%are given a Schnyder wood of a planar triangulation and two weight
%distributions $\mathbf{w}$ and $\mathbf{w}'$ on the set of interior
%faces, then morphing linearly between the corresponding drawings
%defines a planar morph.

%\mchange{Our main result can be stated as follows.

\begin{thm}
\label{thm:main}
Let $T$ be a planar triangulation and let $S$ and $S'$ be two Schnyder
woods of $T$. Let $\Gamma$ and $\Gamma'$ be weighted Schnyder
drawings of $T$ obtained from $S$ and $S'$ together with some weight
distributions. 
% Let $T$ be a planar triangulation and let $S$ and $S'$ be Schnyder
% woods of $T$.  Consider the weighted Schnyder drawings $\Gamma$ and
% $\Gamma'$ obtained from $S$ and $S'$ together with weight
% distributions $\mathbf{w}$ and $\mathbf{w'}$, respectively.
There
exists a sequence of straight-line planar drawings of $T$ $\Gamma =
\Gamma_{0},\ldots,\Gamma_{k+1} = \Gamma'$ such that $k$ is $O(n^{2})$,
the linear morph $\langle \Gamma_{i},\Gamma_{i+1}\rangle$ is planar,
$0\leq i\leq k$, and the vertices of each $\Gamma_{i}$, $1\leq i\leq
k$, lie in \rnote{a $(6n-15) \times (6n-15)$ grid.}  Furthermore,
these drawings can be obtained in polynomial time.
\end{thm}

\rnote{We now describe how the results in the upcoming
  sections prove the theorem.
  Lemma~\ref{lem:morph_weight_distribution}
  (Section~\ref{sec:morph_weight_distribution}) proves that if we
  perform a linear morph between two weighted Schnyder drawings that
  differ only in their weight distribution then planarity is
  preserved.
%As we will see in Lemma~\ref{lem:morph_weight_distribution}
%(Section~\ref{sec:morph_weight_distribution})
%redistributing weight
%defines a planar linear morph and this can be performed in one step. 
Thus, we may take $\Gamma_{1}$ and $\Gamma_{k}$ to be the
drawings obtained from the uniform weight
distribution 
on $S$ and $S'$ respectively.  By Schnyder's
Theorem~\ref{thm:schnyder_embedding} these drawings lie on a $(2n-5)
\times (2n-5)$ grid and we can scale them up to our larger grid.  By
Lemma~\ref{lem:lattice_not_4_connected_sequence}
(Section~\ref{sec:Schnyder})
%(Section~\ref{subsec:flip_triangle}) 
there is a sequence of $k$ flips,
$k \in O(n^{2})$, that converts $S$ to $S'$.
%As shown by Felsner in~\cite{Felsner04}
%flipping cyclically oriented triangles is the basic operation that is
%required to traverse the lattice of Schnyder woods.  Now,
%Lemma~\ref{lem:lattice_not_4_connected_sequence} guarantees the
%existence of a sequence of flips of length $k\coloneqq O(n^{2})$
%between $S$ and $S'$ in the Schnyder wood lattice. 
Therefore it suffices to show that each flip in the sequence can be
realized via a planar morph composed of a constant number of linear
morphs.
%It would suffice to
%show that it is possible to morph linearly using a constant number of
%steps between each intermediate Schnyder wood in the flip sequence. 
%We show that this is indeed the case. 
In Theorem~\ref{thm:morph_face_flip} (Section~\ref{sec:face_morph})
we prove that if we perform a linear morph  
between two weighted Schnyder drawings
that differ only by a flip of a face then planarity is preserved.  
%we show that it is possible to morph linearly, while preserving
%planarity, when flipping a face.
% 
In Theorem~\ref{thm:morph_separating_triangle}
(Section~\ref{sec:separating_triangle_morph}) we prove that if two
Schnyder drawings with the same uniform weight distribution differ by
a flip of a separating triangle then there is a planar morph between
them composed of three linear morphs.  The intermediate drawings
involve altered weight distributions (here
Lemma~\ref{lem:morph_weight_distribution} is used again), and lie on a
grid of the required size.
%we prove that when
%flipping a separating triangle we can morph between the corresponding
%drawings in at most 4 steps. 
%Each step is a linear morph between
%weighted Schnyder drawings and may require some weights to be
%redistributed (here Lemma~\ref{lem:morph_weight_distribution} is used again).
%
Putting these results together gives the final sequence
$\Gamma_{0},\ldots,\Gamma_{k+1}$.  All the intermediate drawings
%in the sequence from $\Gamma_{1}$ to $\Gamma_{k}$
%are weighted Schnyder drawings that are realizable in 
%an $ O(n)\times O(n)$
lie in a $(6n-15) \times (6n-15)$ grid and each of them can be
obtained in $O(n)$ time from the previous one.  This completes the proof of
Theorem~\ref{thm:main} modulo the proofs in the following sections.  }

\remove{
We make one remark about the complications caused by separating
  triangles.  It is tempting to simply combine morphs of the outside
  and the inside of each separating triangle.  Essentially, this
  assigns coordinates to each vertex relative to the Schnyder wood of
  the 4-connected component containing the vertex.  The \Achange{trouble} %difficulty
  with this approach is that separating triangles can be nested
  $c\cdot n$ deep, and we lose control of the grid size. 
  } 

\remove{ \bnote{An alternative to using
    Theorem~\ref{thm:morph_separating_triangle} when flipping a
    separating triangle is to assign to each vertex the coordinates
    relative to the Schnyder wood of the 4-connected block they are
    contained in. Using this technique allows us to perform this kind
    of morph in one step and showing that morphing linearly preserves
    planarity follows from Theorem~\ref{thm:morph_face_flip}. However,
    using this method, the size of the grid may no longer be $
    O(n)\times O(n).$} }

%%%%%%%%%%%%%%%%%%%%%%%%%%%%%%%%%%%%%%%%%
%%%%%%%%%%%%%%%%%%%%%%%%%%%%%%%%%%%%%%%%%
%\section{Planar morphs from weight distributions}
\section{\rfnote{Morphing to change weight distributions}}
\label{sec:morph_weight_distribution}

\remove{
Consider two weighted Schnyder drawings $\Gamma$ and
  $\Gamma'$ that differ only in their weight distributions,
  i.e.,~$\Gamma$ and $\Gamma'$ are obtained from the same graph with
  the same Schnyder wood but with different weight distributions.
  (Recall that weight distributions were defined in
  Section~\ref{sec:Schnyder}.)
  In this section we show that the linear morph from $\Gamma$ to $\Gamma'$ preserves planarity.
  }

%\remove{
\longVerNote{
  In this section we give the first ingredient of our main result.  We
  consider two weighted Schnyder drawings $\Gamma$ and $\Gamma'$ that
  differ only in their weight distributions, i.e.,~$\Gamma$ and
  $\Gamma'$ are obtained from the same graph with the same Schnyder
  wood but with different weight distributions.  (Recall that weight
  distributions were defined in Section~\ref{sec:Schnyder}.)  In the
  following result we show that the linear morph from $\Gamma$ to
  $\Gamma'$ preserves planarity.
}

\begin{lem} %(proof in the long version)%appendix)
\label{lem:morph_weight_distribution}
  Let $T$ be a planar triangulation and let $S$ be a Schnyder wood of
  $T$. Consider two weight distributions $\mathbf{w}$ and
  $\mathbf{w}'$ \pnote{on the faces of $T$}, and denote by $\Gamma$
  and $\Gamma'$ the weighted Schnyder drawings of $T$ obtained from
  $\mathbf{w}$ and $\mathbf{w}'$ respectively.  Then the linear morph
  $\langle\Gamma,\Gamma'\rangle$ is planar.
\end{lem}
\longVerNote{
\begin{proof}%[Proof of Lemma~\ref{lem:morph_weight_distribution}]
  Consider the family of functions $\{\mathbf{w}^{t}\}_{t\in[0,1]}$,
  defined by
  %\[
  $\mathbf{w}^{t}(f) = (1-t) \mathbf{w}(f) + t\mathbf{w}'(f).$
  %\]
  \pnote{Note that for every $t\in[0,1]$, the function
    $\mathbf{w}^{t}$ is a weight distribution since
    $\mathbf{w}^{t}(f)$ is positive for all $f$ and $\sum_f
    \mathbf{w}^{t}(f) = 2n-5$.}
  \pnote{By Dhandapani~\cite{Dhandapani10} each $\mathbf{w}^{t}$
    yields a planar drawing. }  This family of drawings defines a
  \pnote{planar} morph from $\Gamma$ to $\Gamma'$. To conclude the
  proof we \pnote{only need to}
%  are only left to
  show that this morph is linear.
  \pnote{
The
  position of vertex $x$ at time $t$ is
  $x^{t}=(x_{1}^{t},x_{2}^{t},x_{3}^{t})$ where 
%For $i=1,2,3$ the coordinates  of vertex $x$ at time $t$ are given by} 
%Let 
$x_{i}^{t}=\sum_{f\in R_{i}(x)}\mathbf{w}^{t}(f)$.  % and 
%Note that
%  $r_{i}^{t}(x)=(1-t)r_{i}^{0}(x)+tr_{i}^{1}(x)$. 
%  The position of $x$ at time $t$ is  $x^{t}=(r_{1}^{t}(x),r_{2}^{t}(x),r_{3}^{t}(x))$, which can be
%  rewritten as 
Note that $x^{t}= (1-t)x^{0}+tx^{1}$, so the result now follows.}
\end{proof}
}
%%

%%%%%%%%%%%%%%%%%%%%%%%%%%%%%%%%%%%%%%%%%
%%%%%%%%%%%%%%%%%%%%%%%%%%%%%%%%%%%%%%%%%
%\section{Flips on faces define morphs}
\section{\rfnote{Morphing to flip a facial triangle}}
\label{sec:face_morph}

 In this section we prove that the linear morph from one
  Schnyder drawing to another one\wsnote{,} obtained by flipping a
  cyclically oriented face and keeping the same weight distribution,
  preserves planarity.  
  \lAnote{This is Theorem~\ref{thm:morph_face_flip} below.}
  See Figure~\ref{fig:anim_3_steps}.  We begin by showing how the
  regions for each vertex change during such a flip and then we use
  this to show how the
  %barycentric 
  coordinates change.

\begin{figure}
  \centering
  \scalebox{.62}{\input{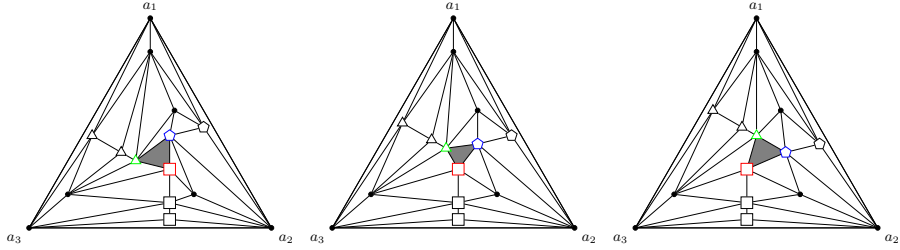} }
  \caption{Snapshots from a linear morph defined by a flip of the shaded face at times $t=0$, $t=0.5$
and $t=1.$ \wsnote{The trajectory of rectangular shaped vertices is parallel
to the exterior edge $a_{2}a_{3}$. Similar properties hold for triangular and pentagonal
shaped vertices.}}
  \label{fig:anim_3_steps}
\end{figure}
% \mchange{Our main result in this section shows that if we are given Schnyder
% woods $S$ and $S'$ that differ by a facial flip, then morphing
% linearly between the corresponding weighted Schnyder drawings
% preserves planarity. First we analyze changes in the regions and
% coordiantes for each vertex from $S$ to $S'$.}

\remove{This section is devoted to the analysis of flips of faces in
  Schnyder woods. We begin by investigating how the regions for each
  vertex change and use this to derive the variation of the
  barycentric coordinates induced by the Schnyder wood. We finalize
  the section by proving that morphing linearly from one Schnyder
  drawing to another one obtained by flipping a cyclically oriented
  face, both drawings given by a fixed weight distribution, defines a
  planar morph.}
%Rewrite

%\rrnote{Here's a space-saving version of the next paragraph:
\rrnote{Let $S$ and $S'$ be Schnyder woods of triangulation $T$ that
  differ by a flip on face $xyz$ oriented counterclockwise in $S$ with
  $(x,y)$ of colour 1.  Let $(v_{1},v_{2},v_{3})$ and
  $(v_{1}',v_{2}',v_{3}')$ be the coordinates of vertex $v$ in the
  weighted Schnyder drawings from $S$ and $S'$ respectively with
  respect to weight distribution $\mathbf{w}$.}  \remove{\rrnote{skip
    this:} Let us begin by introducing some notation. Consider a
  planar triangulation $T$. \mchange{Let $S$ and $S'$ be Schnyder
    woods that differ by a face flip, say face $xyz$. We may assume
    without loss of generality that $f$ is cyclically oriented
    counterclockwise in $S$ and that $(x,y)\in S$ has colour 1.
%equipped with a Schnyder wood $S$. Suppose there is
%a cyclically oriented face $xyz$ in $S$, without loss of generality we
%will assume that $xyz$ is oriented counterclockwise and that $(x,y)$
%has colour 1. Let $S'$ be the Schnyder wood obtained from $S$ by
%flipping $xyz$. 
    Fix a weight distribution $\mathbf{w}$ and let
    $(v_{1},v_{2},v_{3})$ and $(v_{1}',v_{2}',v_{3}')$ denote the
    coordinates for vertex $v$ in the weighted Schnyder drawings from
    $S$ and $S'$ respectively.}}
%Let $\mathbf{w}$ be a weight distribution and let
%$(v_{1},v_{2},v_{3})$ denote the Schnyder coordinates for vertex $v$
%obtained from $S$ and $\mathbf{w}$, and let $(v'_{1},v'_{2},v'_{3})$
%denote the Schnyder coordinates for vertex $v$ obtained from $S'$ and
%$\mathbf{w}$.
%\rrnote{resume here (attach to above para.)} 
For an interior edge $pq$ of $T$, let $\Delta_{i}(pq)$ be the set of
faces \pnote{in the region} bounded by $pq$ and the paths $P_{i}(p)$
and $P_{i}(q)$ in $S$, \nnote{and define $\delta_{i}(pq)$ to be the
  weight of that region,
  i.e.,~$\delta_{i}(pq)=\sum_{f\in\Delta_{i}(pq)}\mathbf{w}(f)$.}
\pnote{We use notation $P_i(v), R_i(v)$, and $D_i(v)$ as defined in
  Section~\ref{sec:Schnyder} and $\Delta_{i}(pq)$ as above and add
  primes to denote the corresponding structures in $S'$.}
%. For an
%interior edge $pq$ we also define
%$\delta_{i}(pq)=\sum_{f\in\Delta_{i}(pq)}\mathbf{w}(f)$.
%The proof of the following Lemma can be found in the Appendix.
%\Achange{The following two lemmas are proved formally in the long version.}
%use the notation above and are proved in the Appendix.}
%
%Let us begin by identifying properties of $S$ and $S'$.  

\lAnote{Figure~\ref{fig:face_flip_coords} shows how the regions change between the weighted Schnyder drawings from  $S$ and $S'$. 
Observe that the outgoing paths $P_2(x)$ and $P_3(x)$ from $x$ do not change, so region $R_1(x)$ is unchanged.  Outgoing path $P_1(x)$ changes, so there are some faces, in particular $\Delta_{1}(yz)\cup\{f\}$, that leave $R_2(x)$ and join $R'_3(x)$.  We capture these properties in the following lemma.}

\begin{lem}\label{lem:sch_woods_before_and_after}
%  Using the notation as above, 
The following conditions hold: % (see Figure~\ref{fig:face_flip_coords}):}
  % Let $T$ be a planar triangulation and let $S$ and $S'$ be Schnyder
  % woods of $T$ such that $S'$ is obtained from $S$ by flipping a
  % counterclockwise cyclically oriented face $xyz$ with $(x,y)$
  % coloured 1. Then the following conditions hold:
  \begin{enumerate}
  \item $R_{1}(x)=R_{1}'(x)$, $R_{3}(y)=R_{3}'(y)$ and
    $R_{2}(z)=R_{2}'(z)$.     
    \pnote{ \item $R'_2(x) = R_2(x) \setminus (\Delta_1(yz) \cup
      \{f\}), R'_3(x) = R_3(x) \cup (\Delta_1(yz) \cup \{f\})$ and
      similarly for $y$ and $z$.}    
  \item $D_{1}(x)=D_{1}'(x)$, $D_{2}(z)=D_{2}'(z)$ and
    $D_{3}(y)=D_{3}'(y)$. 
  \end{enumerate}
\end{lem}
\longVerNote{
\begin{proof}%[Proof of Lemma~\ref{lem:sch_woods_before_and_after}]
 \pnote{We show the first result in each condition since
    the others can be derived similarly.}

  \begin{enumerate}
  \item\label{lem:item:one}  Note
    that the only path outgoing from $x$ that changes is $P_{1}(x)$ (see
    Figure~\ref{fig:face_flip_coords}), that is, $P_{2}(x)=P_{2}'(x)$
    and $P_{3}(x)=P_{3}'(x)$ and therefore $R_{1}(x)=R_{1}'(x)$.

  \bnote{\item Observe that faces in $f\cup\Delta_{1}(yz)$ are to the left of
    $P_{1}(x)$ and therefore $\Delta_{1}(yz)\cup\{f\}\subseteq
    R_{2}(x)$; \rrnote{analogously} $\Delta_{1}(yz)\cup\{f\}\subseteq
    R'_{3}(x)$, see Figure~\ref{fig:face_flip_coords}. It can be seen
    that these are all the faces that change regions, around $x$, from
    $S$ to $S'$. Therefore $R'_2(x) = R_2(x) \setminus (\Delta_1(yz)
    \cup \{f\})$ and $R'_3(x) = R_3(x) \cup (\Delta_1(yz) \cup \{f\})$.}
    
  \item  This is an immediate consequence of
    \ref{lem:item:one}, since no edges of $S$ in $R_{1}(x)$ changed
    from $S$ to $S'$.
    
\remove{  \wsnote{\item\label{item:disjoint_interiors} We prove that the interior of
    $R_{1}(x)$ and $R_{3}(y)$ are pairwise disjoint.
    Since %$(x,y)\in T_{1}$ and
    $y\in R_{3}(x)$ it follows that $R_{3}(y)\subseteq R_{3}(x)$. Now,
    since $R_{1}(x)\cap R_{3}(x)=P_{2}(x)$ we have that $R_{x}(x)\cap
    R_{3}(x)\subseteq P_{2}(x)$. The result now follows.}

  \wsnote{\item We proceed by contradiction,
    suppose there is $u\in D_{1}(x)\setminus\{x\}$ so that $u$ is not
    in the interior of $R_{1}(x)$. By planarity, this implies there is
    $v\in P_{j}(x)\cap (D_{1}(x)\setminus\{x\})$, $j\in\{2,3\}$. This
    implies that $T_{1}^{-}\cup T_{j}^{-}$ contains a cycle, which
    contradicts~\ref{item:path_no_common_vertex}.

    We can now deduce that $D_{1}(x)$ is disjoint from
    $D_{3}(y)$. This follows from~\ref{item:disjoint_interiors} and
    since $x$ is in the interior of $R_{1}(y)$, so $x\not\in
    D_{3}(y)$, and $(x,y)\in P_{1}(x)$, so $y\not\in D_{1}(x)$.}}
  \end{enumerate}
\end{proof}
}

\begin{figure}[!ht]
  \centering  
    \scalebox{.82}{\input{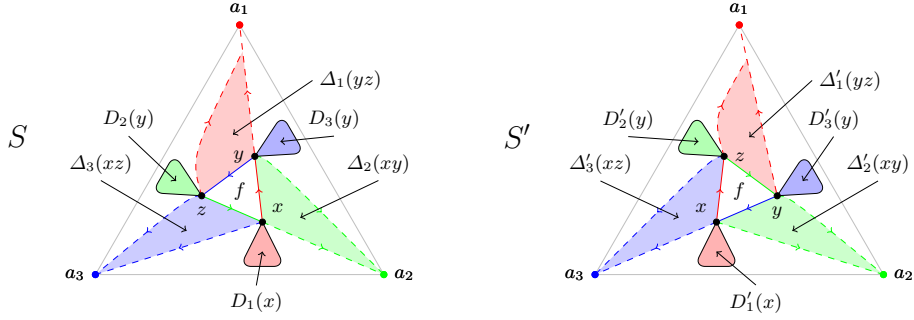}}
  \caption{
    A flip of a counterclockwise oriented face triangle $xyz$ showing changes to the regions. Observe that
    $\Delta_{1}(yz)\cup\{f\}$ leaves $R_2(x)$ and joins $R'_3(x)$. 
    % Changes in the regions of $x$:
    % $\Delta_{1}(yz)\cup\{f\}$ leaves $R_2(x)$ and joins
    %$R'_3(x)$.
  }
  
  % \caption{Changes in the regions of $x$ are given by
  % $\Delta_{1}(yz)\cup\{f\}\subseteq R_{2}(x),R_{3}'(x)$.}
  \label{fig:face_flip_coords}
\end{figure}

Next we study the difference between the coordinates of the weighted
Schnyder drawings corresponding to $S$ and $S'$.
\lAnote{Since the weights do not change, the coordinates of a vertex $v$
    change only if its regions change.  Furthermore, the regions of
    $v$ change only if the paths leaving $v$ change, and a path
    changes only if it uses an edge of $f$.  Thus the only vertices
    whose coordinates change are those in $D_{1}(x)\cup D_{2}(z)\cup
    D_{3}(y)$.   Furthermore, for a vertex in $D_1(x)$ the amount of change is captured by the faces that switch regions.  We make this more precise in the following lemma.}

\begin{lem}\label{lem:face_flip_coordinates}
 \Achange{ For each $v\in V(T)$,} 
%  Using notation as above, the following relation
%  holds for each $v\in V(T)$:
  \[
 (v_{1}',v_{2}',v_{3}')=
  \begin{cases}
    (v_{1},v_{2},v_{3}) \quad\quad\quad  \text{if }v\not\in D_{1}(x)\cup D_{2}(z)\cup  D_{3}(y) \\
    (v_{1},v_{2}-(\delta_{1}(yz)+\mathbf{w}(f)),v_{3}+\delta_{1}(yz)+\mathbf{w}(f))
&  \text{if }v\in  D_{1}(x)\\
    (v_{1}+\delta_{2}(xy)+\mathbf{w}(f),v_{2},v_{3}-(\delta_{2}(xy)+\mathbf{w}(f)))
& \text{if }v\in D_{2}(z)\\
    (v_{1}-(\delta_{3}(xz)+\mathbf{w}(f)),v_{2}+\delta_{3}(xz)+\mathbf{w}(f),v_{3})
& \text{if }v\in D_{3}(y).
  \end{cases}
  \]
%  where $(v_{1},v_{2},v_{3})$ and $(v_{1}',v_{2}',v_{3}')$ denote the
%  coordinates of $v$ in $\Gamma$ and $\Gamma'$ respectively.
\end{lem}
\longVerNote{
\begin{proof}%[Proof of Lemma~\ref{lem:face_flip_coordinates}]
  
  \pnote{%I think this proof could be made more concise:
%    Since the weights do not change, the coordinates of a vertex $v$
%    change only if its regions change.  Furthermore, the regions of
%    $v$ change only if the paths leaving $v$ change, and a path
%    changes only if it uses an edge of $f$.  Thus the only vertices
%    whose coordinates change are those in $D_{1}(x)\cup D_{2}(z)\cup
%    D_{3}(y)$. 
    \lAnote{As mentioned above, the only vertices whose coordinates
      change are those in $D_{1}(x)\cup D_{2}(z)\cup D_{3}(y)$.}
    Consider a vertex $v\in D_{1}(x)$.  (The other cases will follow
    by symmetry.)  From Lemma~\ref{lem:sch_woods_before_and_after}, we
    have $R'_{1}(x)=R_{1}(x)$, $R'_2(x) = R_2(x) \setminus
    (\Delta_1(yz) \cup \{f\})$, and $R'_3(x) = R_3 \cup (\Delta_1(yz)
    \cup \{f\})$.  Therefore $
    (v'_{1},v'_{2},v'_{3})=(v_{1},v_{2}-(\delta_{1}(yz)+\mathbf{w}(f)),v_{3}+\delta_{1}(yz)+\mathbf{w}(f)).$
  } \remove{It is clear that the coordinates do not change for any of
    the exterior vertices. Now, let us analyze how the coordinates
    change for each interior vertex. First observe that the 3 outgoing
    paths from every interior vertex $v$ of $T$ use at most one edge
    in total from $(x,y), (y,z)$ and $(z,x)$, as otherwise this would
    imply that the corresponding paths share a vertex different from
    $v$, thus contradicting
    property~\ref{item:path_no_common_vertex}. Let $v\not\in
    D_{1}(x)\cup D_{2}(z)\cup D_{3}(y)$. Note that its regions remain
    unchanged, since none of the outgoing paths from $v$ changed, and
    \pnote{fix} so do its coordinates, as the weights remain
    fixed. Therefore for such vertex $v$ we have
    $(v_{1}',v_{2}',v_{3}')=(v_{1},v_{2},v_{3})$.}
% If the
% Schnyder coordinates of an interior vertex $v$ are to remain fixed
% from $S$ to $S'$, then it must be the case that its regions also
% remain unchanged. The regions will remain unchanged if and only if
% none of the three paths leaving $v$ use any of the edges $(x,y),
% (y,z)$ and $(z,x)$. In particular the $1$-path does not use the edge
% $(x,y)$ and therefore $v\not\in D_{1}(x)$. Similarly, $v\not\in
% D_{3}(y)$ and $v\not\in D_{2}(z)$. So if then the coordinates of $v$
% remain
% unchanged, in other words, .
%  \begin{figure}[!bt]
%    \centering    
%    \input{figs/descendants.tix}
%    \caption{\rrnote{Please add edges of $R_1(x)$ etc., i.e.~blue edge from $x$ to
%$a_3$, etc.} \pnote{ Some people will print B\&W.  Also, should the da%rk/light
%colours be explained?}
%    Vertices whose regions change from $S$ to $S'$.}
%    \label{fig:movable}
%  \end{figure}
%
\remove{ Now we analyze how the coordinates change from $\Gamma$ to
  $\Gamma'$ for the vertices in $D_{1}(x)\cup D_{2}(z)\cup D_{3}(y)$,
  see Figure~\ref{fig:movable}. By
  Lemma~\ref{lem:sch_woods_before_and_after} we have that the sets
  $D_{1}(x)$, $D_{2}(z)$ and $D_{3}(y)$ are pairwise disjoint and by
  symmetry we may just consider the case $v\in D_{1}(x)$. First
  observe that the face $f$ is contained in $R_{2}(v)$ in $S$, whereas
  in $S'$ it is contained in $R'_{3}(v)$. It is not only $f$ but also
  the faces in $\Delta_{1}(yz)$ that are in $R_{2}(v)$ in $S$, and in
  $R_{3}'(v)$ in $S'$, as shown in
  Figure~\ref{fig:face_flip_coords}. As in
  Lemma~\ref{lem:sch_woods_before_and_after} we have
  $R_{1}(v)=R'_{1}(v)$. These are the only changes in the regions
  corresponding to $v$. Therefore $R'_{1}(v)=R_{1}(v)$,
  $R'_{2}(v)=R_{2}(v)\setminus(\Delta_{1}(yz)\cup\{f\})$ and
  $R'_{3}(x)=R_{3}(x)\cup\Delta_{1}(yz)\cup \{f\}$. This implies that
  the relation between the Schnyder coordinates of $x$ in $\Gamma$ and
  $\Gamma'$ is given by
  \[
  (v'_{1},v'_{2},v'_{3})=(v_{1},v_{2}-(\delta_{1}(yz)+\mathbf{w}(f)),v_{3}+\delta_{1}(yz)+\mathbf{w}(f)).
  \]
  The cases when $v\in D_{2}(x)$ and $v\in D_{3}(x)$ follow from an
  analogous argument.} 
\end{proof} 
}
%%

% Now that we know how the coordinates change from one weighted Schnyder drawing
% to the other, we use this to prove that morphing linearly between the
% drawings does not cause any face in the region $R_{1}(x)$ to collapse. This
% is done in the following lemma.
% 
% \begin{lem}\label{lem:face_flip_region_1}
%   Let $S$ be a Schnyder wood of a planar triangulation $T$ that
%   contains a face $f=xyz$ bounded by a counterclockwise directed
%   triangle with $(x,y)$ coloured $1$,
%   and let $S'$ be the Schnyder wood obtained from $S$ by flipping
%   $f$. Denote by $\Gamma$ and $\Gamma'$ the weighted Schnyder drawings
%   obtained from $S$ and $S'$ respectively and from a fixed weight
%   distribution $\mathbf{w}$. Then during the morph
%   $\langle\Gamma,\Gamma'\rangle$ no face in $R_{1}(x)$ collapses.
% \end{lem}

We are ready
%now in position 
to prove the main result of this section.
\nnote{We express it in terms of a general weight distribution since
  we will need that in the next section.}

\begin{thm}\label{thm:morph_face_flip}
  Let $S$ be a Schnyder wood of a planar triangulation $T$ that
  contains a face $f$ bounded by a counterclockwise directed triangle
  \rrnote{$xyz$}, and let $S'$ be the Schnyder wood obtained from $S$
  by flipping $f$. Denote by $\Gamma$ and $\Gamma'$ the weighted
  Schnyder drawings obtained from $S$ and $S'$ respectively
  \rrnote{with weight distribution $\mathbf{w}$}\wsnote{.}
%  and from a fixed weight distribution. 
  Then $\langle\Gamma,\Gamma'\rangle$ is a planar morph.
\end{thm}
\begin{proof}
  If a triangle collapses during the morph, then it must be incident
  to at least one vertex that moves, i.e.,~one of $D_1(x), D_2(y)$ or
  $D_3(z)$.  By property~\ref{item:region_property3}, apart from
  $x,y,z$ these vertex sets lie in the interiors of regions $R_1(x),
  R_2(y), R_3(z)$ respectively. Thus it suffices to show that no
  triangle in one of these regions collapses, \Bchange{and that no
    triangle incident to $x,y$ or $z$ collapses.

    \longVerNote{ 
      Let $pqr$ be a triangle such that $pqr\in R_{1}(x)$.  (The argument
      for triangles in other regions is similar.) We prove
        that $pqr$ does not collapse by using Corollary~5
        from~\cite{BarreraCruz13}. We restate this corollary as a
        claim here.

    \begin{claim}
      Consider a morph $M$ acting on points $p$, $q$ and $r$ such
      that their motion is along the same direction and at constant
      speed. If $p$ is to the right of the line through $qr$ at the
      beginning and the end of the morph $M$, then $p$ is to the right of
      the line through $qr$ throughout the morph $M$.
    \end{claim}
    
    %Let $t=pqr$ be the face in $R_{1}(x)$.  
    Observe that all of $p, q$ and $r$ move in the same direction,
    namely a direction parallel to the exterior edge
    $a_{2}a_{3}$. This holds since the first coordinate of all three
    vertices remains unchanged during the morph by
    Lemma~\ref{lem:face_flip_coordinates}. Suppose, without loss of
    generality, that $p$ lies to the right of the line through $qr$ in
    $\Gamma$. This must also be the case in $\Gamma'$. %, as the
    %combinatorial embedding of the triangulation is unique.
    Therefore by the Claim above, it follows that $p$ lies to the
    right of $qr$ throughout $\langle\Gamma,\Gamma'\rangle$. In
    particular this implies that $pqr$ does not collapse during the
    morph.  \Bchange{The same argument applies to a face in
      $\Delta_{3}(xz)\cup \Delta_{2}(xy)$ that is incident to $x$ but
      not incident to either $y$ or $z$.}  }

}
 %   Let $t$ be a triangle such that $t\in R_{1}(x)$.} (The argument
  % for triangles in other regions is similar.)  Any vertex of $R_1(x)$
  % that moves is in $D_1(x)$ and by
  % Lemma~\ref{lem:face_flip_coordinates} these vertices are all
  % translated by the same amount.  \Achange{We}
  % %The idea is to 
  % argue that if triangle $b,c,e$ in clockwise order collapses as we
  % translate a subset of its vertices then the end result is triangle
  % $b,c,e$ in counterclockwise order.  This contradicts the fact that
  % $\Gamma$ and $\Gamma'$ have the same faces.  A rigorous proof
  % is %given
  % in the long version. \Bchange{The same argument applies to a triangle in
  % $\Delta_{3}(xz)\cup \Delta_{2}(xy)$ that is incident to
  % $x$ but not incident to either $y$ or $z$.}

  \Bchange{It remains to prove that no triangle $t$ incident to at least two
  vertices of $x,y$ and $z$ collapses. Here we only consider the case
  where $t=xyz$, the other case can be handled similarly.}  We will
  show that $x$ never lies on the line segment $yz$ during the morph.
  (The other two cases are similar.)  Since $(x,y)$ has colour $1$ in
  $S$, it follows that $x\in R_{1}(y)$. Similarly, since $(z,x)$ has
  colour $2$ in $S$, we have that $x\in R_{1}(z)$. Therefore
  $x_{1}<y_{1},z_{1}$. Using a similar argument on $S'$ we obtain that
  $x_{1}'<y_{1}',z_{1}'$.  Finally, note that $x_1 = x'_1$.  This
  implies that $x$ never lies on the line segment $yz$ during the
  morph.
\end{proof}

\MoveToAppendix{\begin{proof}
  \rrnote{If a triangle collapses during the morph it must include at
    least one vertex that moves, i.e.~one of $D_1(x), D_2(y)$ or
    $D_3(z)$.  By Lemma~\ref{lem:sch_woods_before_and_after}, apart
    from $x,y,z$ these vertex sets lie in the interiors of regions
    $R_1(x), R_2(y), R_3(z)$ respectively.  Thus it suffices to show
    that no triangle in one of these regions collapses, and that
    \Bchange{no triangle incident to $x,y$ or $z$ collapses.} 
    %triangle $xyz$ does not collapse.

    Consider a triangle in region $R_1(x)$.  

    \rfnote{It remains to prove that} \Bchange{no triangle incident to
      $x,y$ or $z$ collapses. Consider a triangle $t$ different from
      $xyz$ incident to $x$ and not in $R_{1}(x)$. We may assume
      without loss of generality that $t\in\Delta_{3}(xz)$, the case
      $t\in\Delta_{2}(xy)$ can be handled similarly. Suppose
      $t=xw_{1}w_{2}$, with $w_{1},w_{2}\in \Delta_{3}(xz).$ Note that
      any vertex in $w\in\Delta(xz)_{3}\setminus\{x,z\}$ has its third
      coordinate greater than that of $x$ and $z$ throughout the
      morph, since $x,z\in R_{3}(w)$ and $x,z\in R_{3}'(w)$. This
      prevents triangle $xw_{1}w_{2}$ from collapsing during the
      morph. Finally, let us show that}
%Here we only show that }
%To complete the proof we must show that
  triangle $xyz$ does not collapse.
% and \Bchange{we show no other
%    triangle incident to $x,y$ or $z$ collapses in the Appendix}.
  We will show that $x$ never lies on the line segment $yz$ during the
  morph.  (The other two cases are similar.)  Since $(x,y)$ has colour
  $1$ in $S$, it follows that $x\in R_{1}(y)$. Similarly, since
  $(z,x)$ has colour $2$ in $S$, we have that $x\in
  R_{1}(z)$. Therefore $x_{1}<y_{1},z_{1}$. Using a similar argument
  on $S'$ we obtain that $x_{1}'<y_{1}',z_{1}'$.  Finally, note that
  $x_1 = x'_1$.  This implies that $x$ never lies on the line segment
  $yz$ during the morph.  }
 \end{proof}}
 \remove{ Let
   $\{\Gamma^{t}\}_{t\in[0,1]}=\langle\Gamma,\Gamma'\rangle$. Clearly
   $\Gamma=\Gamma^{0}$ and $\Gamma'=\Gamma^{1}$ are planar
   drawings. We are only left to show that $\Gamma^{t}$ is planar for
   all $t\in(0,1)$. Let us assume, without loss of generality, that
   $f=xyz$ with $(x,y)$ having colour 1 in $S$. Note that the set of
   descendants of $x$, $D_{1}(x)$, remains unchanged from $S$ to $S'$,
   and also recall that $R_{1}(x)=R'_{1}(x)$, by
   Lemma~\ref{lem:sch_woods_before_and_after}. First let us observe
   that among the vertices that change position, see
   Lemma~\ref{lem:face_flip_coordinates}, those in $D_{1}(x)$ remain
   in the region of the plane bounded by $x,x'$,$P_{3}(x)$ and
   $P_{2}(x)$, see Figure~\ref{fig:region2}. This follows from the
   fact that for $t=0$ and $t=1$ each of these vertices belongs to
   $R_{1}(x)$ and that the trajectory followed by each vertex is
   monotone. This implies that if any face collapses, then the face
   must belong to one of $R_{1}(x)$, $R_{2}(z)$ or $R_{3}(y)$, or it
   must be $xyz$.
  \begin{figure}[!ht] 
    \centering
    \input{figs/flip_movement.tix}
    \caption{\rrnote{This figure has become irrelevant to the proof.  Let's delete
it.} Region in which $x$ and its descendants move.}
    \label{fig:region2}
  \end{figure}
%  It follows from Lemma~\ref{lem:face_flip_region_1} that no face
%  collapses in $R_{i}(x)$ during the morph, here $i\in\{1,2,3\}$.
  \change{Observe that no face in $R_{i}(x)$ collapses during the morph,
    as otherwise this would need to be undone by $t=1$ and thus
    monotonicity of the trajectories would be contradicted. }
    \pnote{The level of detail is skewed -- too much detail from most
      of the proof and too little detail here.} 
    
    We are only left to discard one last possible source of faces
    collapsing, namely $xyz$. Let us show that there is no point in
    time $t\in[0,1]$ such that a vertex of $xyz$, say $x$, lies in the
    line segment determined by the other two points, in this case
    $yz$. Since $(x,y)$ has colour $1$ in $S$, it follows that $x\in
    R_{1}(y)$. Similarly, since $(z,x)$ has colour $2$ in $S$, we have
    that $x\in R_{1}(z)$. Therefore $x_{1}<y_{1},z_{1}$. \bnote{Using
      a similar argument on $S'$ we obtain that
      $x_{1}'<y_{1}',z_{1}'$.
%Now, in $S'$ we
%  have that $(x,z)$ has colour $1$ and $(y,x)$ has colour $3$. This
%  implies $x\in R_{1}'(z)$ and $x\in R_{1}'(y)$. Hence
%  $x_{1}'<y_{1}',z_{1}'$.
  Thus we have $x_{1}<y_{1},z_{1}$ and
  $x_{1}'<y_{1}',z_{1}'$ and since $x_{1}=x_{1}'$ then
% at times $0$ and $1$ respectively. Now, since
% $x_{1}=x_{1}'$ it follows that for any $t\in(0,1)$
% and $s\in[0,1]$
%  we will have 
%  $x_{1}<s((1-t)y_{1}+ty_{1}')+(1-s)((1-t)z_{1}+tz_{1}')$. Thus
  $x^{t}$ cannot be in the line segment joining $z^{t}$ and
  $y^{t}$.} The result now follows.
%\end{proof}
}%
%%%%%%%%%%%%%%%%%%%%%%%%%%%%%%%%%%%%%%%%%
%%%%%%%%%%%%%%%%%%%%%%%%%%%%%%%%%%%%%%%%%
%\section{Morphs from flips in separating triangles}
\section{\rfnote{Morphing to flip a separating triangle}}
\label{sec:separating_triangle_morph}
In this section we prove that there is a planar morph between any two
weighted Schnyder drawings that differ by a separating triangle flip \lAnote{(Theorem~\ref{thm:morph_separating_triangle} below)}.
%As mentioned in Section~\ref{sec:main}, morphing
%linearly might not preserve planarity.  
Our morph will be composed of three
linear morphs. %\longVerNote{We decompose the morph into three linear morphs since
%morphing linearly between drawings may not preserve planarity}

\begin{figure}[h]  
  \centering
  \scalebox{0.5}{\input{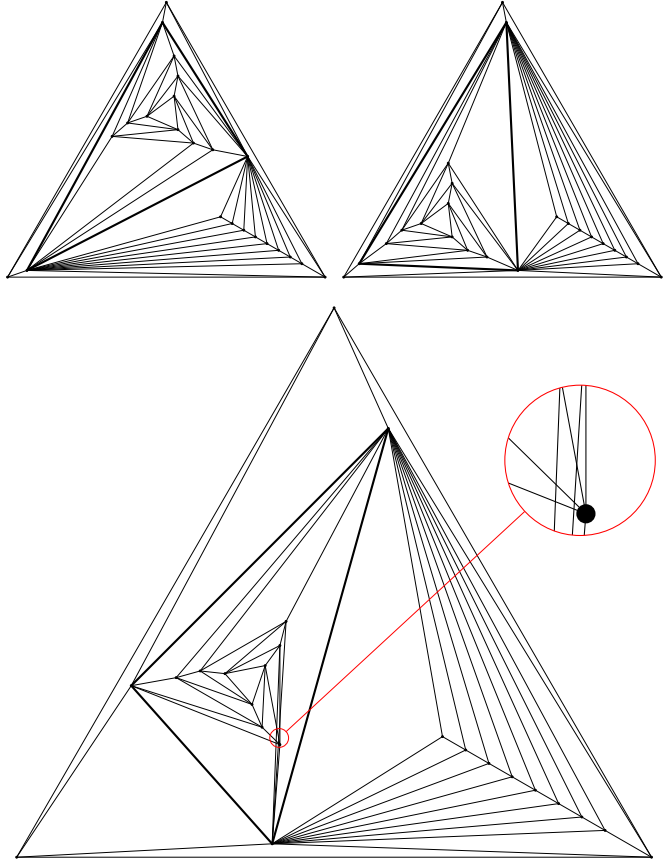}}
  \caption{\wsnote{The linear morph defined by a flip of a separating
      triangle might not be planar if weights are not distributed
      appropriately. Here we illustrate the flip of the separating
      triangle in thick edges. Snapshots at $t=0$, and $t=1$ are the
      top. The bottom drawing corresponds to $t=0.7$; note the edge
      crossings.}}
  \label{fig:bad_separating_triangle_flip} 
\end{figure} 
%And delete the next sentence:}
%As we saw in section~\ref{subsec:flip_triangle}, given a Schnyder wood
%with a cyclically oriented triangle, it is possible to flip such
%triangle to obtain a second Schnyder wood. In the previous section we
%showed that when such triangle is a face, we can obtain a planar
%linear morph between the two weighted Schnyder drawings. 
% \mchange{In this section we show that if we are given Schnyder woods
%   $S$ and $S'$ that differ by a separating triangle flip, then
%   morphing linearly between the corresponding weighted Schnyder
%   drawings preserves planarity.}
% %In this
% %section we obtain a similar result for the case where the triangle is
% %a separating triangle, the difference is that
% \change{As mentioned in Section~\ref{sec:main} morphing linearly between the
% corresponding drawings might not preserve planarity.} 
%
%Let us begin by introducing some notation. 
Throughout this section we
let $S$ and $S'$ be Schnyder woods of a planar triangulation $T$ such
that $S'$ is obtained from $S$ after flipping a counterclockwise
oriented separating triangle $C=xyz$, with $(x,y)$ coloured $1$ in
$S$.
%Consider the uniform weight
%distribution $\mathbf{u}$ and the two weighted Schnyder drawings
%$\Gamma$ and $\Gamma'$ obtained from $S$ and $\mathbf{u}$ and from
%$S'$ and $\mathbf{u}$ respectively.
Let $\Gamma$ and $\Gamma'$ be two weighted Schnyder drawings obtained
from $S$ and $S'$ respectively with weight distribution $\mathbf{w}$.
For the main result of the section, it suffices to consider a uniform
weight distribution because we can get to it via a single planar
linear morph, as shown in Section~\ref{sec:morph_weight_distribution}.
However, for the intermediate results of the section we need more
general weight distributions.

We now give an outline of the
strategy we follow. 
%I suggest a higher level idea of what will happen, e.g.:
%The trouble with 
Morphing linearly from $\Gamma$ to $\Gamma'$ may cause faces inside
$C$ to collapse.  \longVerNote{An example is provided in Figure~\ref{fig:bad_separating_triangle_flip}.}
However, we can show that there is a ``nice'' weight distribution that
prevents this from happening.  Our plan, therefore, is to morph
linearly from $\Gamma$ to a drawing $\overline{\Gamma}$ with a nice
weight distribution, then morph linearly to drawing
$\overline{\Gamma}'$ to effect the separating triangle flip.  A final
change of weights back to the uniform distribution gives a linear
morph from $\overline{\Gamma}'$ to $\Gamma'$.
%And then delete rest of para.}
\remove{To morph between $\Gamma$ and $\Gamma'$ we will
require to first morph to intermediate drawings, that is, we morph
linearly from $\Gamma$ to $\overline{\Gamma}$, then we morph linearly
from $\overline{\Gamma}$ to $\overline{\Gamma}'$ and finally we morph
linearly from $\overline{\Gamma}'$ to $\Gamma'$. Here
$\overline{\Gamma}$ and $\overline{\Gamma}'$ are weighted Schnyder
drawings obtained from $S$ and $S'$ respectively together with some
special weight distribution $\mathbf{w}$.}

This section is structured as follows. First we study how the
coordinates change between 
$\Gamma$ and $\Gamma'$. 
%weighted Schnyder drawings obtained from $S$ and $S'$ and some
%weight distribution. 
Next we show that faces strictly interior to $T|_{C}$ do not collapse
during a linear morph between $\overline{\Gamma}$ and
$\overline{\Gamma}'$.  We then give a similar result for faces of
$T|_{C}$ that share a vertex or edge with $C$ \emph{provided that} the
weight distribution satisfies certain properties.  Finally we prove
the main result by showing that there is a weight distribution with
the required properties.% And delete rest of para.}
% weighted Schnyder drawings obtained from $S$ and $S'$ and some
% weight distribution. We then show that no interior face of $T|_{C}$
% collapses while morphing linearly between $\overline{\Gamma}$ and
% $\overline{\Gamma}'$, where these drawings are defined in terms
% of a special weight distribution, see
% Lemma~\ref{lem:flip_separating}. After this, we show that the morph
% $\langle\overline{\Gamma},\overline{\Gamma}'\rangle$ is planar. We
% then prove the main result of this section, which is the following.

Let us begin by examining the coordinates of vertices.  For
  vertex $b \in V(T)$ let $(b_{1},b_{2},b_{3})$ and
  $(b'_{1},b'_{2},b'_{3})$ denote its coordinates in $\Gamma$ and
  $\Gamma'$ respectively. 
%writing down the coordinates in $\Gamma$ for interior
%vertices of $T|_{C}$. 
For $b$ an interior vertex of $T|_{C}$
%and denote by 
let $\beta_{i}$ be the $i$-th coordinate of $b$ in $T|_{C}$
when considering the restriction of $S$ to $T|_{C}$ with weight
distribution $\mathbf{w}$.  By analyzing Figure~\ref{fig:b_coords}, we
can see that the coordinates for $b$ in $\Gamma$ are
\begin{equation}\label{eq:coordinates_interior_interior_vertex}
\begin{split}
(b_{1},b_{2},b_{3})&=(x_{1}+\delta_{3}(xz)+\beta_{1},z_{2}+\delta_{1}(yz)+\beta_{2},y_{3}+\delta_{2}(xy)+\beta_{3})\\
  &=(x_{1},z_{2},y_{3})+(\delta_{3}(xz),\delta_{1}(yz),\delta_{2}(xy))+(\beta_{1},\beta_{2},\beta_{3}).
\end{split}
\end{equation}

We now analyze how the coordinates of vertices change from $\Gamma$ to
$\Gamma'$.  We use % \rnote{-- how about $\mathbf{w}_C$ -- }
$\mathbf{w}_{C}$ to denote the weight of faces inside $C$,
i.e.,~$\mathbf{w}_{C}=\sum_{f\in \Fc(T|_{C})}\mathbf{w}(f)$.
\lAnote{Note that reading the proof of the lemma first will make it easier to understand the formulas stated in the lemma.}

% Let us begin by writing down the coordinates in $\Gamma$ for interior
% vertices of $T|_{C}$. Let $b$ be an interior vertex of $T|_{C}$ and
% denote by $\beta_{i}$ the $i$-th coordinate of $b$ in $T|_{C}$ when
% considering the restriction of $S$ to $T|_{C}$ with weight
% distribution $\mathbf{w}$. By analyzing Figure~\ref{fig:b_coords}, we
% can see that the coordinates for $b$ in $\Gamma$ are
% \begin{equation}\label{eq:coordinates_interior_interior_vertex}
% \begin{split}
% (b_{1},b_{2},b_{3})&=(x_{1}+\delta_{3}(xz)+\beta_{1},z_{2}+\delta_{1}(yz)+\beta_{2},y_{3}+\delta_{2}(xy)+\beta_{3})\\
%   &=(x_{1},z_{2},y_{3})+(\delta_{3}(xz),\delta_{1}(yz),\delta_{2}(xy))+(\beta_{1},\beta_{2},\beta_{3}),
% \end{split}
% \end{equation}
% We now analyze how the coordinates of vertices change from $\Gamma$ to
% $\Gamma'$.

\begin{lem}% (proof in the long version)
\label{lem:sep_triangle_flip_coordinates}
\Achange{For each $b \in V(T)$,}
%Using the notation as above the following relation holds for each $b \in V(T)$:
  \[
  (b_{1}',b_{2}',b_{3}')=
  \begin{cases}
    (b_{1},b_{2}-(\delta_{1}(yz)+\mathbf{w}_{C}),b_{3}+\delta_{1}(yz)+\mathbf{w}_{C}) & \text{if }b\in D_{1}(x)\\
    (b_{1}+\delta_{2}(xy)+\mathbf{w}_{C},b_{2},b_{3}-(\delta_{2}(xy)+\mathbf{w}_{C})) & \text{if }b\in D_{2}(z)\\
    (b_{1}-(\delta_{3}(xz)+\mathbf{w}_{C}),b_{2}+\delta_{3}(xz)+\mathbf{w}_{C},b_{3}) &
    \text{if }b\in D_{3}(y)\\
    (x_{1},z_{2},y_{3})+(\delta_{2}(xy),\delta_{3}(xz),\delta_{1}(yz))+(\beta_{3},\beta_{1},\beta_{2}) &
    %\text{if }b\text{ is an interior vertex of }T|_{C}\\
    \text{if }b \in {\cal I}\\
    (b_{1},b_{2},b_{3}) &\text{otherwise}
  \end{cases} 
  \] 
  \Achange{where ${\cal I}$ is the set of interior vertices of $T|_{C}$.}
  % where $\mathbf{w}(T|_{C})=\sum_{f\in \Fc(T|_{C})}\mathbf{w}(f)$, and
  % $(\beta_{1},\beta_{2},\beta_{3})$ denotes the coordinates of $b$ in
  % $T|_{C}$ with respect to $\mathbf{w}$ and $S|_{C}$.
\end{lem}
\longVerNote{
\begin{proof}%[Proof of Lemma~\ref{lem:sep_triangle_flip_coordinates}]
  %Consider a planar triangulation $T$ with a cyclically oriented
  %separating triangle $C$ in a Schnyder wood $S$. Let $S'$ be the
  %Schnyder wood obtained from $S$ by flipping $C$.
  Observe that the coordinates of a vertex $b$ change only if
    its regions change, and its regions change only if \Achange{an outgoing path from}
    $b$ %\Bchange{in $S$} 
    uses an interior edge of $T|_{C}$ or an edge of
    $C$. Therefore the only vertices whose coordinates change are the
    interior vertices of $T|_{C}$ or vertices in $D_{1}(x)\cup
    D_{2}(z)\cup D_{3}(y)$.  \rnote{The part of the result for $b \in
      D_{1}(x)\cup D_{2}(z)\cup D_{3}(y)$ follows from
      Lemma~\ref{lem:face_flip_coordinates} applied to $T\setminus C$
      with the restrictions of the Schnyder woods $S$ and $S'$
      to %the set of interior edges of
      $T\setminus C$ and the weight distribution where the weight of
      the face $xyz$ is equal to $\mathbf{w}_{C}$ and the weight of
      the remaining faces is given by $\mathbf{w}$.  }
 
    \remove{ Let us assume that $b\in D_{1}(x)\cup D_{2}(z)\cup
      D_{3}(y)$. Consider $T\setminus C$, the two restrictions of $S$
      and $S'$ to the set of interior edges of $T\setminus C$ and a
      weight distribution for $T\setminus C$ where the weight of the
      face $xyz$ is equal to $\mathbf{w}_{C}$ and the weight of the
      remaining faces is given by $\mathbf{w}$. The first part of the
      result now follows from Lemma~\ref{lem:face_flip_coordinates}
      applied to $T\setminus C$ considering the Schnyder woods and
      weight distribution described above.  }

  Now suppose $b$ is an interior vertex of $T|_{C}$. The coordinates
  of $b$ in $T|_{C}$ are given
  by~\eqref{eq:coordinates_interior_interior_vertex}.
  
  Similarly %to~\eqref{eq:coordinates_interior_interior_vertex}
  (see Figure~\ref{fig:b_coords}) the coordinates for $b$ in $\Gamma'$ are
  given by
  \[
%\begin{split}
  (b_{1}',b_{2}',b_{3}')%=(x_{1}+\delta_{2}(xy)+\beta'_{1},z_{2}+\delta_{3}(xz)+\beta'_{2},y_{3}+\delta_{1}(yz)+\beta'_{3})
  =(x_{1},z_{2},y_{3})+(\delta_{2}(xy),\delta_{3}(xz),\delta_{1}(yz))+(\beta_{1}',\beta_{2}',\beta_{3}'),
%\end{split}
\]
where $\beta'_{i}$ denotes the $i$-th coordinate of $b$ in
$T|_{C}$. Finally, 
%note that 
since the colours of the interior edges of $T|_{C}$ change from $i$ to
$i+1$ thus $(\beta_{1}',\beta_{2}',\beta_{3}') =
  (\beta_{3},\beta_{1},\beta_{2})$, which gives the required formula.
\end{proof}
}
%%
%\Achange{See the Appendix for the proof.}
  \begin{figure}
    \centering 
\scalebox{.82}{    %\usetikzlibrary{decorations.markings,calc}

\begin{tikzpicture} [scale=0.9,
  _vertex/.style ={circle,draw=black, fill=black,inner sep=1pt},
  r_vertex/.style={circle,draw=red,  fill=red,inner sep=1pt},
  g_vertex/.style={circle,draw=green,fill=green,inner sep=1pt},
  b_vertex/.style={circle,draw=blue, fill=blue,inner sep=1pt},
  _edge/.style={black,line width=1pt},
  r_edge/.style={red,line width=0.3pt},
  g_edge/.style={green,line width=0.3pt},
  b_edge/.style={blue,line width=0.3pt},
 every edge/.style={draw=black,line width=0.3pt}]

  \begin{scope}[decoration={ markings, mark=at position 0.5 with {\arrow{>}}}]
    \def\rad{3}
    
    \coordinate (r) at (90:\rad);
    \coordinate (g) at (330:\rad);
    \coordinate (b) at (210:\rad);

    \begin{scope}[yshift=2.5cm]
      \coordinate (incoming_red1) at (-110:1.5);
      \coordinate (incoming_red2) at (-90:1.5);
      \coordinate (incoming_red3) at (-70:1.5);
    \end{scope}

    \begin{scope}[xshift=-1.5cm,yshift=-1cm]
      \coordinate (incoming_blue1) at (15:1);
      \coordinate (incoming_blue2) at (45:1); 
      \coordinate (incoming_blue3) at (75:1);
    \end{scope}
    
    \begin{scope}[xshift=1.5cm,yshift=-1cm]
    \coordinate (incoming_green1) at (120:1);
    \coordinate (incoming_green2) at (140:1);
    \coordinate (incoming_green3) at (160:1);
    \coordinate (incoming_green4) at (180:1);
    \end{scope}
    % Exterior black triangle
    \draw[gray!50] (r)--(g)--(b)--cycle;
    % \draw (r)--(g)--(b)--cycle;

    % The 4 vertices
    \node[r_vertex,label=above:$a_{1}$] (ar) at (r) {};
    \node[g_vertex,label=east:$a_{2}$] (ag) at (g) {};
    \node[b_vertex,label=west:$a_{3}$] (ab) at (b) {};
%    \node[_vertex] (u) at (0,0) {};

    % The positions for the interior vertices
    \coordinate (x) at (barycentric cs:r=4,g=9,b=6);
    \coordinate (y) at (barycentric cs:r=9,g=6,b=4);
    \coordinate (z) at (barycentric cs:r=6,g=4,b=9);
    \coordinate (r_path) at ($(y)!0.8!(r)$);

    % \begin{scope}[shift=(x)]
    %   \fill[red,rounded corners=5pt] (0,0)-- (-0.5,-0.8) --
    %   (0.5,-0.8) -- (0,0);
    %   \draw[rounded corners=5pt] (0,0)-- (-0.5,-0.8) --
    %   (0.5,-0.8) -- (0,0);
    % \end{scope}

    % Delta 1
    \fill[red!20]  (z).. controls +(-0.2,0.7)
    ..(r_path.west) -- (y) -- (z);% -- (z); 

    \draw[r_edge,dashed,postaction={decorate}] (z).. controls
    +(-0.2,0.7) ..(r_path.west);

	\coordinate (Delta1p1) at (z);
	\coordinate (Delta1p2) at (y);
	\coordinate (Delta1p3) at (r_path);
	\coordinate (DeltaLabel1) at (barycentric cs:Delta1p1=1,Delta1p2=1,Delta1p3=1);
	\node (Delta1Label) at (2,2) {$\Delta_{1}(yz)$};
	\draw[->] (Delta1Label.west) -- (DeltaLabel1);

    % Delta 2
    \fill[green!20]  (y).. controls +(0.2,0)
    ..(g) -- (x) -- (y);% -- (z); 

    \draw[g_edge,dashed,postaction={decorate}]
    (y).. controls +(0.2,0)
    ..(g);
    \draw[g_edge,dashed,postaction={decorate}] (x) -- (g);% -- (z);     

	\coordinate (Delta2p1) at (x);
	\coordinate (Delta2p2) at (y);
	\coordinate (Delta2p3) at (g);
	\coordinate (DeltaLabel2) at (barycentric cs:Delta2p1=1,Delta2p2=1,Delta2p3=1);
	\node (Delta2Label) at (2.5,0.5) {$\Delta_{2}(xy)$};
	\draw[->] (Delta2Label.south) -- (DeltaLabel2);

    % Delta 3
    \fill[blue!20]  (z).. controls +(-0.2,0)
    ..(b) -- (x) -- (z);% -- (z); 

    \draw[b_edge,dashed,postaction={decorate}]
    (z).. controls +(-0.2,0)
    ..(b);
    \draw[b_edge,dashed,postaction={decorate}] (x) -- (b);% -- (z);     

	\coordinate (Delta3p1) at (x);
	\coordinate (Delta3p2) at (z);
	\coordinate (Delta3p3) at (b);
	\coordinate (DeltaLabel3) at (barycentric cs:Delta3p1=1,Delta3p2=1,Delta3p3=1);
	\node (Delta3Label) at (-2.5,0.5) {$\Delta_{3}(xz)$};
	\draw[->] (Delta3Label.south) -- (DeltaLabel3);

        \node (SchnLabel) at (-4,1) {{\Large$S$}};

    % The three interior vertices
    \node[_vertex,label=below:$x$] (xv) at (x) {};
    \node[_vertex,label=above right:$y$] (yv) at (y) {};
    \node[_vertex,label=above left:$z$] (zv) at (z) {};

    \node[_vertex,label=right:$b$] (bv) at (0,0) {};

    \draw[r_edge,dashed,postaction={decorate}] (bv) -- (yv);
    \draw[b_edge,dashed,postaction={decorate}] (bv) -- (zv);
    \draw[g_edge,dashed,postaction={decorate}] (bv) -- (xv);

    %\node (f) at (0,0) {$f$};

    %\node[rectangle,dashed,draw=red,fill=red!20,minimum size=3mm,label=right:$\Delta_{1}(yz)$] (Delta) at (1.5,2) {};
    %\node[_vertex,label=below:$rp$] (rp) at (r_path) {};

    % Interior flippable triangle
    \draw[r_edge,postaction={decorate}] (xv) -- (yv);
    \draw[b_edge,postaction={decorate}] (yv) -- (zv);
    \draw[g_edge,postaction={decorate}] (zv) -- (xv);    

    % Incoming paths at exterior vertices.
    % blue
%    \draw[b_edge,dashed,postaction={decorate}] (zv) -- (ab);
    % red
    \draw[r_edge,dashed,postaction={decorate}] (yv) -- (ar);
    % green
 %   \draw[g_edge,dashed,postaction={decorate}] (xv) -- (ag);
    % The 4 vertices
    \node[r_vertex,label=above:$a_{1}$] (ar) at (r) {};
    \node[g_vertex,label=east:$a_{2}$] (ag) at (g) {};
    \node[b_vertex,label=west:$a_{3}$] (ab) at (b) {};

  \end{scope}

  \begin{scope}[decoration={ markings, mark=at position 0.5 with {\arrow{>}}},xshift=9cm]
    \def\rad{3}
    
    \coordinate (r) at (90:\rad);
    \coordinate (g) at (330:\rad);
    \coordinate (b) at (210:\rad);

    \begin{scope}[yshift=2.5cm]
      \coordinate (incoming_red1) at (-110:1.5);
      \coordinate (incoming_red2) at (-90:1.5);
      \coordinate (incoming_red3) at (-70:1.5);
    \end{scope}

    \begin{scope}[xshift=-1.5cm,yshift=-1cm]
      \coordinate (incoming_blue1) at (15:1);
      \coordinate (incoming_blue2) at (45:1); 
      \coordinate (incoming_blue3) at (75:1);
    \end{scope}
    
    \begin{scope}[xshift=1.5cm,yshift=-1cm]
    \coordinate (incoming_green1) at (120:1);
    \coordinate (incoming_green2) at (140:1);
    \coordinate (incoming_green3) at (160:1);
    \coordinate (incoming_green4) at (180:1);
    \end{scope}
    % Exterior black triangle
    \draw[gray!50] (r)--(g)--(b)--cycle;
    % \draw (r)--(g)--(b)--cycle;

    % The 4 vertices
    \node[r_vertex,label=above:$a_{1}$] (ar) at (r) {};
    \node[g_vertex,label=east:$a_{2}$] (ag) at (g) {};
    \node[b_vertex,label=west:$a_{3}$] (ab) at (b) {};
%    \node[_vertex] (u) at (0,0) {};

    % The positions for the interior vertices
    \coordinate (x) at (barycentric cs:r=4,g=6,b=9);
    \coordinate (y) at (barycentric cs:r=6,g=9,b=4);
    \coordinate (z) at (barycentric cs:r=9,g=4,b=6);
    \coordinate (r_path) at ($(y)!0.8!(r)$);

    % \begin{scope}[shift=(x)]
    %   \fill[red,rounded corners=5pt] (0,0)-- (-0.5,-0.8) --
    %   (0.5,-0.8) -- (0,0);
    %   \draw[rounded corners=5pt] (0,0)-- (-0.5,-0.8) --
    %   (0.5,-0.8) -- (0,0);
    % \end{scope}

    % Delta 1
    \fill[red!20]  (z).. controls +(-0.2,0.7)
    ..(r_path.west) -- (y) -- (z);% -- (z); 

    \draw[r_edge,dashed,postaction={decorate}] (z).. controls
    +(-0.2,0.7) ..(r_path.west);

	\coordinate (Delta1p1) at (z);
	\coordinate (Delta1p2) at (y);
	\coordinate (Delta1p3) at (r_path);
	\coordinate (DeltaLabel1) at (barycentric cs:Delta1p1=1,Delta1p2=1,Delta1p3=1);
	\node (Delta1Label) at (2,2) {$\Delta_{1}'(yz)$};
	\draw[->] (Delta1Label.west) -- (DeltaLabel1);

    % Delta 2
    \fill[green!20]  (y).. controls +(0.2,0)
    ..(g) -- (x) -- (y);% -- (z); 

    \draw[g_edge,dashed,postaction={decorate}]
    (y).. controls +(0.2,0)
    ..(g);
    \draw[g_edge,dashed,postaction={decorate}] (x) -- (g);% -- (z);     

	\coordinate (Delta2p1) at (x);
	\coordinate (Delta2p2) at (y);
	\coordinate (Delta2p3) at (g);
	\coordinate (DeltaLabel2) at (barycentric cs:Delta2p1=1,Delta2p2=1,Delta2p3=1);
	\node (Delta2Label) at (2.5,0.5) {$\Delta_{2}'(xy)$};
	\draw[->] (Delta2Label.south) -- (DeltaLabel2);

    % Delta 3
    \fill[blue!20]  (z).. controls +(-0.2,0)
    ..(b) -- (x) -- (z);% -- (z); 

    \draw[b_edge,dashed,postaction={decorate}]
    (z).. controls +(-0.2,0)
    ..(b);
    \draw[b_edge,dashed,postaction={decorate}] (x) -- (b);% -- (z);     

	\coordinate (Delta3p1) at (x);
	\coordinate (Delta3p2) at (z);
	\coordinate (Delta3p3) at (b);
	\coordinate (DeltaLabel3) at (barycentric
        cs:Delta3p1=1,Delta3p2=1,Delta3p3=1);

        \node (SchnLabel) at (-4,1) {{\Large$S'$}};
	\node (Delta3Label) at (-2.5,0.5) {$\Delta_{3}'(xz)$};
	\draw[->] (Delta3Label.south) -- (DeltaLabel3);

    % The three interior vertices
    \node[_vertex,label=below:$x$] (xv) at (x) {};
    \node[_vertex,label=above right:$y$] (yv) at (y) {};
    \node[_vertex,label=above left:$z$] (zv) at (z) {};

    \node[_vertex,label=left:$b$] (bv) at (0,0) {};
    \draw[r_edge,dashed,postaction={decorate}] (bv) -- (zv);
    \draw[b_edge,dashed,postaction={decorate}] (bv) -- (xv);
    \draw[g_edge,dashed,postaction={decorate}] (bv) -- (yv);    

    %\node (f) at (0,0) {$f$};

    %\node[rectangle,dashed,draw=red,fill=red!20,minimum size=3mm,label=right:$\Delta_{1}(yz)$] (Delta) at (1.5,2) {};
    %\node[_vertex,label=below:$rp$] (rp) at (r_path) {};

    % Interior flippable triangle
    \draw[r_edge,postaction={decorate}] (xv) -- (zv);
    \draw[b_edge,postaction={decorate}] (yv) -- (xv);
    \draw[g_edge,postaction={decorate}] (zv) -- (yv);    

    % Incoming paths at exterior vertices.
    % blue
%    \draw[b_edge,dashed,postaction={decorate}] (zv) -- (ab);
    % red
    \draw[r_edge,dashed,postaction={decorate}] (yv) -- (ar);
    % green
 %   \draw[g_edge,dashed,postaction={decorate}] (xv) -- (ag);
    % The 4 vertices
    \node[r_vertex,label=above:$a_{1}$] (ar) at (r) {};
    \node[g_vertex,label=east:$a_{2}$] (ag) at (g) {};
    \node[b_vertex,label=west:$a_{3}$] (ab) at (b) {};

  \end{scope}

\end{tikzpicture} }         
    \caption{A flip of a counter-clockwise oriented separating triangle $xyz$.
    %The cyclically oriented separating triangle $xyz$ in $S$ (left) and in $S'$ (right).
    }
    \label{fig:b_coords}
  \end{figure}
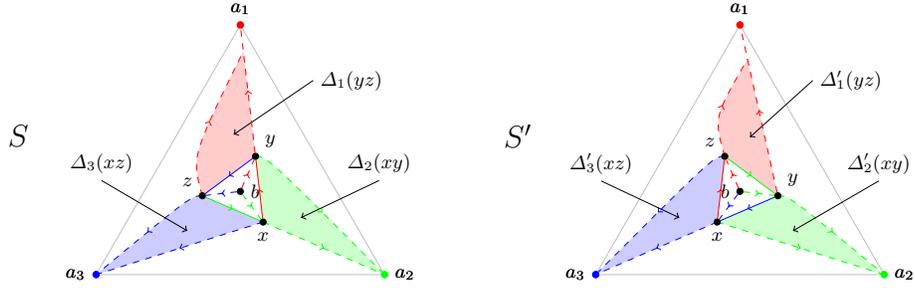

We now examine what happens during a linear morph from $\Gamma$ to
$\Gamma'$.  We first deal with faces strictly interior to $C$.
%The following two lemmas are proved formally in the long version.
%use the notation above and are proved formally in the Appendix.

\begin{lem}\label{lem:sep_triangle_morph_strict_interior}
%\nnote{For general weight distribution -- because in this case we really need that.  Also, say ``using notation as above''.} 
  For an arbitrary weight distribution %and using notation as above we have that 
  no face formed by interior vertices of $T|_{C}$ collapses
  in the morph $\langle\Gamma,\Gamma'\rangle$.
\end{lem}
\begin{proof}%[Proof sketch]
%  A formal %algebraic
%  proof can be found in the Appendix, but we offer the following
%  intuition:
%  \longVerNote{
%  We begin by providing a high level idea of the proof in the next
%  paragraph. Then we present the detailed proof.}

\lAnote{The idea of the proof is as follows.}
Consider a face inside $C$ formed by internal vertices $b,c,e$ whose
coordinates with respect to $T|_C$ are $\beta, \gamma, \varepsilon$,
respectively.
Examining~\eqref{eq:coordinates_interior_interior_vertex} and
Lemma~\ref{lem:sep_triangle_flip_coordinates} we see that the
coordinates of $b,c,e$ in $\Gamma$ and $\Gamma'$ depend in exactly the
same way on the parameters from $T \setminus C$ and differ only in the
parameters $\beta, \gamma, \varepsilon$.  Therefore triangle $bce$
collapses during the morph if and only if it collapses during the
linear transformation on $\beta, \gamma, \varepsilon$ where we perform
a cyclic shift of coordinates, viz., $(\beta_1, \beta_2, \beta_3)$
becomes $(\beta_3, \beta_1, \beta_2)$, etc.  No triangle collapses
during this transformation because it corresponds to moving each of
the three outer vertices $x,y,z$ in a straight line to its clockwise
neighbour.
%More algebraic details can be found in the Appendix.

\lAnote{We now give algebraic details.}
\longVerNote{
  Let $b,c,e\in V(T|_{C})$ be interior vertices of $T|_{C}$ such that
  $bce$ is an interior face of $T|_{C}$. We proceed by contradiction
  by assuming there is a time $r\in(0,1)$ during the linear morph such
  that $bce$ collapses, say $b^{r}$ is in the line segment joining
  $c^{r}$ and $e^{r}$. That is, assume that
  \begin{equation}
    b^{r}=se^{r}+(1-s)c^{r}\label{eq:b_collinear}
  \end{equation}
  for some $r\in(0,1)$ and
  $s\in[0,1]$. By~\eqref{eq:coordinates_interior_interior_vertex} and
  Lemma~\ref{lem:sep_triangle_flip_coordinates}, the left hand side
  of~\eqref{eq:b_collinear} can be written as
  \[
  (x_{1},z_{2},y_{3})+(1-r)(\delta_{3}(xz),\delta_{1}(yz),\delta_{2}(xy))+r(\delta_{2}(xy),\delta_{3}(xz),\delta_{1}(yz))+\beta^{r},
  \]
  where
  $\beta^{r}=(1-r)(\beta_{1},\beta_{2},\beta_{3})+r(\beta_{3},\beta_{1},\beta_{2}).$
  Similar to what we have above, by
  using~\eqref{eq:coordinates_interior_interior_vertex} and
  Lemma~\ref{lem:sep_triangle_flip_coordinates}, we can rewrite the
  right hand side of~\eqref{eq:b_collinear} as
  \[
  \begin{aligned}
  (x_{1},z_{2},y_{3})&+(1-r)(\delta_{3}(xz),\delta_{1}(yz),\delta_{2}(xy))\\
  &+r(\delta_{2}(xy),\delta_{3}(xz),\delta_{1}(yz))+s\varepsilon^{r}+(1-s)\gamma^{r},\end{aligned}
  \]
  where $\varepsilon^{r}$ and $\gamma^{r}$ are defined analogously to
  $\beta^{r}$. So, equation~\eqref{eq:b_collinear} is equivalent to
  \[
  \beta^{r}-((1-s)\varepsilon^{r}+s\gamma^{r})=(0,0,0).
  \]
  This can be rewritten as
  \[
  (1-r)\beta^{0}+r\beta^{1}-((1-s)((1-r)\varepsilon^{0}+r\varepsilon^{1})+s((1-r)\gamma^{0}+r\gamma^{1}))=(0,0,0),
  \]
  and rearranging terms yields
  \[
  (1-r)(\beta^{0}-((1-s)\varepsilon^{0}+s\gamma^{0}))+r(\beta^{1}-((1-s)\varepsilon^{1}+s\gamma^{1}))=(0,0,0).
  \]
  This is equivalent to the following system of equations
  \begin{align*}
    %\label{eq:cyclic_system}
      (1-r)(\beta_{1}-((1-s)\varepsilon_{1}+s\gamma_{1}))+r(\beta_{3}-((1-s)\varepsilon_{3}+s\gamma_{3}))&=0\\
      (1-r)(\beta_{2}-((1-s)\varepsilon_{2}+s\gamma_{2}))+r(\beta_{1}-((1-s)\varepsilon_{1}+s\gamma_{1}))&=0\\
      (1-r)(\beta_{3}-((1-s)\varepsilon_{3}+s\gamma_{3}))+r(\beta_{2}-((1-s)\varepsilon_{2}+s\gamma_{2}))&=0.
  \end{align*}
  To simplify the following arguments, we let
  $Q_{i}=\beta_{i}-((1-s)\varepsilon_{i}+s\gamma_{i})$. So the system
  of equations now becomes
  \begin{align}
      (1-r)Q_{1}+rQ_{3}&=0\label{eq:sys_1}\\
      (1-r)Q_{2}+rQ_{1}&=0\label{eq:sys_2}\\
      (1-r)Q_{3}+rQ_{2}&=0.\label{eq:sys_3}
  \end{align}
  Since these coordinates were obtained from a weighted Schnyder drawing, we know
  there is $i\in\{1,2,3\}$ so that $\beta_{i}>\gamma_{i},\varepsilon_{i}$ by 
  \nnote{Property~\ref{item:region_property}}.
% by  Theorem~\ref{thm:schnyder_barycentric_embedding}.
  We may assume without loss of generality that $i=1$. This implies in
  particular that $Q_{1}>0$. Now, we have that
\[\left|
\begin {array}{ccc}
(1-r)&0&r\\
%\noalign{\medskip}
r&(1-r)&0\\
0&r&(1-r)
\end {array}
\right|=3r^{2}-3r+1>0.
\]
Therefore it must be the case that $Q_{i}=0$, $i=1,2,3$. This contradicts
$Q_{1}>0$, so the result now follows.}
\end{proof}

% \rnote{Since we are not giving full proofs anyway, I suggest
% combining the following two lemmas and also being more general than
% for $xy$ and $x$ respectively.  I also suggest a proof outline.}

Next we consider the faces interior to $C$ that share an edge or
vertex with $C$.  We show that no such face
% In the following lemma we show that no interior face of $T|_{C}$
% incident to an exterior edge of $T|_{C}$
collapses, provided that the
%given 
weight distribution $\mathbf{w}$ satisfies
\rfnote{$\delta_{1} = \delta_{2} = \delta_{3}$
where we use
%$\delta_{1}(yz)=\delta_{2}(xy)=\delta_{3}(xz)$.
%For simplicity, let us use
$\delta_{1}$, $\delta_{2}$ and $\delta_{3}$
to denote $\delta_{1}(yz)$, $\delta_{2}(xy)$ and $\delta_{3}(xz)$
respectively. % in what remains of this section.
}

\begin{lem}\label{lem:sep_triangle_morph_ext_edge}
  \bnote{ Let $\mathbf{w}$ be a weight distribution for the interior
    faces of $T$ such that $\delta_{1}=\delta_{2}=\delta_{3}$.  No
    interior face of $T|_{C}$ incident to an exterior vertex of
    $T|_{C}$ collapses during $\langle\Gamma,\Gamma'\rangle$.}
%\begin{lem}\label{lem:sep_triangle_morph_inc_ext_vertex}
  % Consider a weight distribution $\mathbf{w}$ for the interior faces
  % of $T$ such that $\delta_{1}=\delta_{2}=\delta_{3}$. Let $\Gamma$
  % and $\Gamma'$ be the two weighted Schnyder drawings obtained from $S$
  % and $\mathbf{w}$, and from $S'$ and $\mathbf{w}$ respectively.  Then
  % during the linear morph $\langle\Gamma,\Gamma'\rangle$ no face
  % incident to two interior vertices of $T|_{C}$ and $x$ collapses.
\end{lem}

%\rnote{
\begin{proof}%[Proof]
%  We examine separately the case where the interior face is incident
%  to the edge $xy$ of $C$ and the case where the interior face is
%  \Bchange{only} incident to the vertex $x$ of $C$.  The other cases
%  follow by analogous arguments.
%\longVerNote{The following paragraph provides an outline of the proof of this
%  lemma. The detailed proof is provided afterwards.}
 \lAnote{The idea of the proof is as follows.}
  The cases where the interior face is incident to the edge $xy$ of
  $C$ and where the interior face is \Bchange{only} incident to the
  vertex $x$ of $C$ have to be examined separately \lAnote{using} %by 
  similar
  arguments. Consider the case of an interior face $bxy$ incident to
  edge $xy$.  Suppose by contradiction that at time $r \in [0,1]$
  during the morph the face collapses with $b^r$ lying on segment $x^r
  y^r$, say $b^r = (1-s)x^r + s y^r$ for some $s \in [0,1]$.  We use
  formula~\eqref{eq:coordinates_interior_interior_vertex} and
  Lemma~\ref{lem:sep_triangle_flip_coordinates} to re-write this
  equation.  Some further algebraic manipulations (details to follow)
  show that there is no solution for $r$.  The remaining cases follow
  by analogous arguments.

\lAnote{We now turn to the details.}
\longVerNote{
  \bnote{We may assume, without loss of generality, that the interior
    face we are considering is incident to vertex $x$.  We have two
    possible cases.}
    \begin{description}
      \bnote{    \item[Case 1:] The interior face is incident with an exterior edge.
    \item[Case 2:] The interior face is incident with $x$ and two
      interior vertices.}
    \end{description}

    \bnote{In each case we formulate algebraically the fact that the face
    in question collapses and then proceed by contradiction.}

    \begin{description}
      \bnote{\item[Case 1:] Without loss of generality we may assume
        that the exterior edge that is incident to the interior face
        is $xy$, say $bxy$ is the face in question.} Assume, by
      contradiction, that there is a time $r\in(0,1)$ during the morph
      such that $b$, $x$ and $y$ are collinear. That is
      $b^{r}=(1-s)x^{r}+sy^{r}$. Since $(b,x)$ has colours $2$ and $3$
      in $S$ and $S'$ respectively, it follows that
      $x_{1}^{0}<b_{1}^{0}$ and
      $x_{1}^{0}=x_{1}^{1}<b_{1}^{1}$. Similarly
      $y_{3}<b_{3}^{0},b_{3}^{1}$.  Therefore for any $r\in[0,1]$, the
      first coordinate of $b$ is greater than $x_{1}$ and the third
      coordinate of $b$ is greater than $y_{3}$. Therefore $s\in
      (0,1)$, as otherwise at least one of these conditions would not
      hold.

  We write the equation from above explicitly,
  using~\eqref{eq:coordinates_interior_interior_vertex} and Lemma~\ref{lem:sep_triangle_flip_coordinates}:
  \begin{multline*}
     (x_{1},z_{2},y_{3})+  (1-r)(\delta_{3},\delta_{1},\delta_{2})+ r(\delta_{2},\delta_{3},\delta_{1})+ (1-r)(\beta_{1},\beta_{2},\beta_{3})+ r(\beta_{3},\beta_{1},\beta_{2})\\=
    (1-s)\Big[ (x_{1},z_{2},y_{3})+ (1-r)(0,\delta_{3}+\delta_{1},\delta_{2})+ r(0,\delta_{3},\delta_{2}+\delta_{1})\\+ (1-r)(0,\mathbf{w}_{C},0)+ r(0,0,\mathbf{w}_{C})\Big]\\
    +s\Big[ (x_{1},z_{2},y_{3})+ (1-r)(\delta_{2}+\delta_{3},\delta_{1},0)+ r(\delta_{2},\delta_{1}+\delta_{3},0)\\+ (1-r)(\mathbf{w}_{C},0,0)+ r(0,\mathbf{w}_{C},0)\Big],
  \end{multline*}
    % &(x_{1},z_{2},y_{3})+ &(1-r)(\delta_{3},\delta_{1},\delta_{2})+&r(\delta_{2},\delta_{3},\delta_{1})+&(1-r)(\beta_{1},\beta_{2},\beta_{3})+&r(\beta_{3},\beta_{1},\beta_{2})=\\
    % (1-s)\Big[&(x_{1},z_{2},y_{3})+&(1-r)(0,\delta_{3}+\delta_{1},\delta_{2})+&r(0,\delta_{3},\delta_{2}+\delta_{1})+&(1-r)(0,\mathbf{w}_{C},0)+&r(0,0,\mathbf{w}_{C})\Big]\\
    % +s\Big[&(x_{1},z_{2},y_{3})+&(1-r)(\delta_{2}+\delta_{3},\delta_{1},0)+&r(\delta_{2},\delta_{1}+\delta_{3},0)+&(1-r)(\mathbf{w}_{C},0,0)+&r(0,\mathbf{w}_{C},0)\Big],

  here $\mathbf{w}_{C}\coloneqq
  \sum_{f\in\Fc(T|_{C})}\mathbf{w}(f)$. By simplifying and analyzing each
  coordinate separately we obtain the following system of equations
  \begin{align}
    (1-r) (\delta+\beta_{1}-s(2\delta+\mathbf{w}_{C})) + r(\delta+\beta_{3}-s\delta) &= 0\label{eq:coord1}\\
    (1-r) (\beta_{2}-(1-s)(\delta+\mathbf{w}_{C})) + r(\beta_{1}-s(\delta+\mathbf{w}_{C})) &= 0\label{eq:coord2}\\
    (1-r) (\delta+\beta_{3}-(1-s)\delta) + r(\delta+\beta_{2}-(1-s)(2\delta+\mathbf{w}_{C})) &= 0,\label{eq:coord3}
  \end{align}
  where $\delta=\delta_{1}=\delta_{2}=\delta_{3}$. Now, since we have that
  $\delta+\beta_{3}-s\delta>0$ then from \eqref{eq:coord1} we conclude that
  $\delta+\beta_{1}-s(2\delta+\mathbf{w}_{C})<0$. Therefore
  \begin{equation}
    \beta_{1}-s(\delta+\mathbf{w}_{C})<-(1-s)\delta<0\label{eq:b_part2}.
  \end{equation}
  A similar analysis on equation~\eqref{eq:coord3} yields
  \begin{equation}
    \beta_{2}-(1-s)(\delta+\mathbf{w}_{C})<-s\delta<0\label{eq:b_part3}.
  \end{equation}
  It now follows from inequalities~\eqref{eq:b_part2}
  and~\eqref{eq:b_part3} that there is no $r\in[0,1]$
  satisfying~\eqref{eq:coord2}. \bnote{Therefore the interior face incident to
  $xy$ does not collapse.
%\end{proof}

\item[Case 2:]
    Let us now consider the case where the interior face is
    incident to $x$ and two interior vertices, say $b$ and $c$.}
%\begin{proof}[Proof of Lemma~\ref{lem:sep_triangle_morph_inc_ext_vertex}]
  We proceed by contradiction. Let us assume that the face $xbc$
  collapses at time $r$, with $c\in R_{1}(b)$ and $b\in R_{3}(c)$ in
  $S$. This implies in particular that $\gamma_{1}<\beta_{1}$ and that
  $\beta_{3}<\gamma_{3}$.

  We will show that $b$ cannot be in the line segment $xc$ at any time
  $r$ during the morph. So we have $b^r=(1-s)x^{r}+sc^{r}$, with $r\in(0,1)$
  and $s\in(0,1)$. By an analogous argument, it will follow that $c$
  cannot be in the line segment $xb$.
  % Warning, details <By an analogous argument, it will....>

  We write the previous equation
  explicitly. From~\eqref{eq:coordinates_interior_interior_vertex} and
  Lemma~\ref{lem:sep_triangle_flip_coordinates} we obtain:
  \begin{multline*}
     (x_{1},z_{2},y_{3})+(1-r)(\delta_{3},\delta_{1},\delta_{2})+r(\delta_{2},\delta_{3},\delta_{1})+(1-r)(\beta_{1},\beta_{2},\beta_{3})+r(\beta_{3},\beta_{1},\beta_{2})\\=
    (1-s)\Big[ (x_{1},z_{2},y_{3})+     (0,         \delta_{3},\delta_{2})+(1-r)(0,\delta_{1}+\mathbf{w}_{C},0)\\+r(0,0,\delta_{1}+\mathbf{w}_{C})\Big]\\
    +s\Big[ (x_{1},z_{2},y_{3})+(1-r)(\delta_{3},\delta_{1},\delta_{2})+r(\delta_{2},\delta_{3},\delta_{1})\\+(1-r)(\gamma_{1},\gamma_{2},\gamma_{3})+r(\gamma_{3},\gamma_{1},\gamma_{2})\Big]
  \end{multline*}

  % \[
  % \begin{split}
  %   &(x_{1},z_{2},y_{3})+(1-r)(\delta_{3},\delta_{1},\delta_{2})+r(\delta_{2},\delta_{3},\delta_{1})+(1-r)(\beta_{1},\beta_{2},\beta_{3})+r(\beta_{3},\beta_{1},\beta_{2})=\\
  %   (1-s)(&(x_{1},z_{2},y_{3})+     (0,         \delta_{3},\delta_{2})+(1-r)(0,\delta_{1}+\mathbf{w}_{C},0)+r(0,0,\delta_{1}+\mathbf{w}_{C}))\\
  %   +s(&(x_{1},z_{2},y_{3})+(1-r)(\delta_{3},\delta_{1},\delta_{2})+r(\delta_{2},\delta_{3},\delta_{1})+(1-r)(\gamma_{1},\gamma_{2},\gamma_{3})+r(\gamma_{3},\gamma_{1},\gamma_{2}))
  % \end{split}
  % \]
  
  By simplifying and writing equations for each coordinate we obtain
  the following system of equations
  \begin{align}
    (1-r) (\delta+\beta_{1}-s(\delta+\gamma_{1})) + r(\delta+\beta_{3}-s(\delta+\gamma_{3})) &= 0\label{eq:ncoord1}\\
    (1-r) (\beta_{2}-(1-s)(\delta+\mathbf{w}_{C})-s\gamma_{2}) + r(\beta_{1}-s\gamma_{1}) &= 0\label{eq:ncoord2}\\
    (1-r) (\beta_{3}-s\gamma_{3}) + r(\beta_{2}-(1-s)(\delta+\mathbf{w}_{C})-s\gamma_{2}) &= 0,\label{eq:ncoord3}
  \end{align}
  where $\delta=\delta_{1}=\delta_{2}=\delta_{3}$. Now, since
  $\beta_{1}>\gamma_{1}$ from~\eqref{eq:ncoord1} we get that
  $\delta+\beta_{1}-s(\delta+\gamma_{1})>0$ and therefore
  $\delta+\beta_{3}-s(\delta+\gamma_{3})<0$. From the previous
  inequality we get $\beta_{3}-s\gamma_{3}<-(1-s)\delta<0$. Now, using
  this in equation~\eqref{eq:ncoord3} we get that
  $\beta_{2}-(1-s)(\delta+\mathbf{w}_{C})-s\gamma_{2}>0$. Finally, using the
  previous inequality in~\eqref{eq:ncoord2} we obtain that
  $\beta_{1}-s\gamma_{1}<0$, contradicting our original assumption
  that $\beta_{1}>\gamma_{1}$. So the result follows.  
  \end{description}}
\end{proof}
%}
\remove{ We are only left to discard one more possible source of faces
  collapsing within $T|_{C}$, namely faces incident to two interior
  vertices of $T|_{C}$ and an exterior vertex of $T|_{C}$. As in the
  previous lemma, we require that
  $\delta_{1}=\delta_{2}=\delta_{3}$. We prove this for the case when
  the exterior vertex of $T|_{C}$ is $x$, since the other two cases
  follow from a analogous argument.}%
We are now ready to prove the main result of this section.

\begin{thm}\label{thm:morph_separating_triangle}
  Let $T$ be a planar triangulation and let $S$ and $S'$ be two
  Schnyder woods of $T$ such that $S'$ is obtained from $S$ by
  flipping a counterclockwise cyclically oriented separating triangle
  $C=xyz$ in $S$.  Let $\Gamma$ and $\Gamma'$ be weighted Schnyder
  drawings obtained from $S$ and $S'$, respectively, with uniform
  weight distribution.  Then there exist weighted Schnyder drawings
  $\overline{\Gamma}$ and $\overline{\Gamma}'$ on 
  \rnote{a $(6n-15) \times (6n-15)$} 
%an  $ O(n)\times  O(n)$ 
  integer grid such that each of the following linear morphs
  is planar: $\langle \Gamma, \overline{\Gamma}\rangle , \langle
  \overline{\Gamma},\overline{\Gamma}'\rangle,$ and $\langle
  \overline{\Gamma}',\Gamma' \rangle$.
  % Let $T$ be a planar triangulation and let $S$ and $S'$ be two
  % Schnyder woods of $T$ such that $S'$ is obtained from $S$ by
  % flipping a counterclockwise cyclically oriented separating triangle
  % $C=xyz$ in $S$.  Then there exist planar drawings of $T$ in an
  % $ O(n)\times  O(n)$ integer grid $\Gamma_{0},\ldots,\Gamma_{k}$
  % such that $\langle \Gamma_i, \Gamma_{i+1} \rangle$, $0\leq i\leq
  % k-1$, defines a planar morph, where $\Gamma_{0}$ and $\Gamma_{k}$
  % are the weighted Schnyder drawings obtained from $S$ and $S'$ with uniform
  % weight distribution, and $k\leq 3$.
\end{thm}
\begin{proof}%[Proof of Theorem~\ref{thm:morph_separating_triangle}]
  Our aim is to define the planar drawings $\overline\Gamma$ and
  $\overline\Gamma'$.  Each one will be realized in a grid that is
  three times finer than the $(2n-5)\times(2n-5)$ grid, i.e.,~in a
  $(6n-15)\times(6n-15)$ grid with weight distributions that sum to
  $6n-15$.  Under this setup, the initial uniform weight distribution
  $\mathbf{u}$ takes a value of 3 in each interior face.

  Drawings $\overline\Gamma$ and $\overline\Gamma'$ will be the
  weighted Schnyder drawings obtained from $S$ and $S'$ respectively
  with a new weight distribution $\mathbf{\overline w}$.  We use
  $\Delta_{1}$, $\Delta_{2}$ and $\Delta_{3}$ to denote the regions
  $\Delta_{1}(yz)$, $\Delta_{2}(xy)$ and $\Delta_{3}(xz)$
  respectively, in $S$.  We use $\delta_i$ and $\overline\delta_i$ to
  denote the weight of $\Delta_i, i=1,2,3$ with respect to the uniform
  weight distribution and the new weight distribution
  $\mathbf{\overline w}$, respectively.

  We will define $\mathbf{\overline w}$ so that $\overline\delta_{1},
  \overline\delta_{2},$ and $\overline\delta_{3}$ all take on the
  average value $\delta\coloneqq
  (\delta_{1}+\delta_{2}+\delta_{3})/3$.  
  The idea is to remove weight from faces in a region of above-average weight, and add weight to faces in a region of below-average weight.
%  The idea is to remove weight
%  from the faces of any region $\Delta_i$ whose weight is more than
%  the average, and add weight to the faces of any region whose weight
%  is less than the average. 
The new face weights must be positive integers. 
%  We need to show that the new face weights
%  are positive integers.  
  Note first that $\delta$ is an integer.  Note secondly that $\delta
  > \delta_i/3$ for any $i$ since the other $\delta_j$'s are positive.
  Thus $\delta_i - \delta < \frac{2}{3} \delta_i$.  This means that we
  can reduce $\delta_i$ to the average $\delta$ without removing more
  than 2 weight units from any face (of initial weight 3) in any
  region.  There is more than one solution for $\mathbf{\overline w}$,
  but
  the morph might look best if $\mathbf{\overline w}$ is as uniform as
  possible.  To be more specific, we can define new face weights
  $\mathbf{\overline w}$ via the following algorithm: Initialize
  $\mathbf{\overline w}= \mathbf{w}$.  While some $\overline\delta_i$
  is greater than the average $\delta$, remove 1 from a maximum weight
  face of $\Delta_i$ and add 1 to a minimum weight face in a region
  $\Delta_j$ whose weight is less than the average.

  This completes the description of $\overline\Gamma$ and
  $\overline\Gamma'$.  It remains to show that the three linear morphs
  are planar.  
  \rfnote{
  The morphs $\langle \Gamma, \overline\Gamma \rangle$
  and $\langle \overline\Gamma', \Gamma' \rangle$
  only involve changes to the weight distribution so they are planar by 
  Lemma~\ref{lem:morph_weight_distribution}.} 
  %The changes from $\Gamma$ to $\overline\Gamma$ and from
  %$\overline\Gamma'$ to $\Gamma'$ are only changes in weight
  %distribution.  Thus by Lemma~\ref{lem:morph_weight_distribution}
  %those two linear morphs are planar.  
%
  Consider the linear morph
  $\langle \overline\Gamma, \overline\Gamma' \rangle$.  The two
  drawings differ by a flip of a separating triangle.  They have the
  same weight distribution $\mathbf{\overline w}$ which satisfies
  $\overline\delta_{1}=\overline\delta_{2}=\overline\delta_{3}$.  By
  Lemmas~\ref{lem:sep_triangle_morph_strict_interior},
  and~\ref{lem:sep_triangle_morph_ext_edge} no interior face of
  $T|_{C}$ collapses during the morph.  By
  Theorem~\ref{thm:morph_face_flip} no face of $T\setminus C$
  collapses during the morph. Thus %the morph
  $\langle\Gamma,\Gamma'\rangle$ defines a planar morph.
\end{proof}

\section{Identifying weighted Schnyder drawings}
\label{sec:finding-weights}

\remove{ %%% Here is the referee's comment
  In fact, the recent work [1] implies that for any given point set P
  (with triangular convex hull), there is a unique triangulation for
  which P gives a Schnyder drawing (see [*] below). In my opinion this
  makes the contribution of this paper a nice result but a very
  particular case that may be of interest to only a small audience.

  [1] Nicolas Bonichon, Cyril Gavoille, Nicolas Hanusse, and David
  Ilcinkas. Connections between theta-graphs, delaunay triangulations,
  and orthogonal surfaces. In WG'2010, volume 6410 of Lecture Notes in
  Computer Science, pages 266Ð278, 2010.

  [*] Assume that we are given a Schnyder drawing on a set P of
  points. Then that triangulation (endowed with the Schnyder wood) is
  exactly the oriented half-\Theta_6-graph of P as defined in
  [1]. This easily follows from the following (well-known) property of
  Schnyder drawings: If we draw the three lines through a point x \in
  P that are parallel to the three sides of the convex hull of the
  drawing, we obtain 6 sectors and the orientation and color (in the
  Schnyder drawing) of an edge incident to x is entirely determined by
  which cones it lies in.  }

\Achange{ In this section we give a polynomial time algorithm to test
  if a given straight-line planar drawing $\Gamma$ of triangulation
  $T$ is a weighted Schnyder drawing.  The first step is to identify
  the Schnyder wood.  A recent result of Bonichon et
  al.~\cite{Bonichon10} shows that, given a point set $P$ with
  triangular convex hull, a Schnyder drawing on $P$ is exactly the
  ``half-$\Theta_6$-graph'' of $P$, which can be computed efficiently.
  Thus, given drawing $\Gamma$, we first ignore the edges and compute
  the half-$\Theta_6$ graph of the points.  If this differs from
  $\Gamma$, we do not have a weighted Schnyder drawing.  Otherwise,
  the half-$\Theta_6$ graph determines the Schnyder wood $S$.  We next
  find the face weights.}  \remove{ We do not know a polynomial time
  algorithm to test if a given straight-line planar drawing is a
  weighted Schnyder drawing.  However, the problem becomes easy if
  both the drawing and the Schnyder wood are specified.  Let $S$ be a
  Schnyder wood of triangulation $T$ and let $\Gamma$ be a
  straight-line planar drawing of $T$ in the plane $x+y+z=2n-5$ with
  the three outer vertices on the positive axes.  }  We claim that
%It can be shown that
there exists a unique assignment of (not necessarily positive) weights
$\mathbf{w}$ on the faces of $T$ such that $\Gamma$ is precisely the
drawing obtained from $S$ and $\mathbf{w}$ as described
in~(\ref{eq:weighted_drawing}).
Furthermore, %$\mathbf{w}$ is unique and
$\mathbf{w}$ can be found in polynomial time by solving a system of
linear equations in the $2n-5$ variables $\mathbf{w}(f), f \in
\Fc(T)$.  The equations are those from~(\ref{eq:weighted_drawing}).
The rows of the coefficient matrix are the characteristic vectors of
$R_i(v), i \in \{1,2,3\}, v$ an interior vertex of $T$,
%  The matrix of coefficients has rows $\{\chi_{R_{i}(v)}:
%  i\in\{1,2,3\}, v\text{ is an interior vertex of }T\}$,
and the system of equations has a solution because the matrix has rank
$2n-5$.  This was proved by Felsner and Zickfeld~\cite[Theorem
9]{Felsner08}.  (Note that their theorem is about coplanar orthogonal
surfaces; however, their proof considers the exact same set of
equations and their Claims 1 and 2 give the needed result.)

\longVerNote{We conclude this section by providing an example of a planar triangulation and a
  drawing of it such that there is no weight distribution (having only
  positive weights) that
  realizes it. We claim that the drawing shown in
  Figure~\ref{fig:eg_neg_weights} is such an example. Note that
  $R_{1}(x)\subseteq R_{1}(y)$ in any Schnyder wood $S$, independent
  of whether $(y,x)$ or $(y,z)$ belongs to $T_{3}$ in $S$. If weights
  were non-negative then we would have $x_{1}\leq y_{1}$, which is
  clearly not the case here.} 

\begin{figure}[h]
  \centering
  \usetikzlibrary{shadows,decorations}
\usetikzlibrary{decorations.markings,calc}
\begin{tikzpicture} [
  _vertex/.style ={circle,draw=black, fill=black,inner sep=1pt},
  r_vertex/.style={circle,draw=red,  fill=red,inner sep=1pt},
  g_vertex/.style={circle,draw=green,fill=green,inner sep=1pt},
  b_vertex/.style={circle,draw=blue, fill=blue,inner sep=1pt},
  _edge/.style={black,line width=1pt},
  r_edge/.style={red,line width=0.3pt},
  g_edge/.style={green,line width=0.3pt},
  b_edge/.style={blue,line width=0.3pt},
 every edge/.style={draw=black,line width=0.3pt},scale=0.25]

% Coordinates
\begin{scope}[scale=1]
\coordinate (v1) at (-1.22cm, 12.70cm);
\coordinate (v2) at (9.15cm, -4.05cm);
\coordinate (v3) at (-11.30cm, -4.37cm);
\coordinate (v4) at (-0.61cm, 0.71cm);
\coordinate (v5) at (0.50cm, 2.86cm);
\coordinate (v6) at (-0.66cm, 2.28cm);
\coordinate (v7) at (6.51cm, -2.38cm);
\end{scope}
% Vertices
\node[_vertex,label=above:$a_1$](V1) at (v1) {};
\node[_vertex,label=below right:$a_2$](V2) at (v2) {};
\node[_vertex,label=below left:$a_3$](V3) at (v3) {};
\node[_vertex,label=below:$x$](V4) at (v4) {};
\node[_vertex](V5) at (v5) {};
\node[_vertex,label=above left:$z$](V6) at (v6) {};
\node[_vertex,label=above:$y$](V7) at (v7) {};

\begin{scope}[xshift=3.5cm,decoration={ markings, mark=at position 0.5 with {\arrow{>}}}]

% Edges
\draw[thick] (V6) edge (V4);
\draw[thick] (V5) edge (V6);
\draw[thick] (V4) edge (V7);
\draw[r_edge,postaction={decorate}] (V5) -- (V1);
\draw[r_edge,postaction={decorate}] (V6) -- (V1);
\draw[b_edge,postaction={decorate}] (V4) -- (V3);
\draw[b_edge,postaction={decorate}] (V6) -- (V3);
\draw[g_edge,postaction={decorate}] (V4) -- (V2);
\draw[g_edge,postaction={decorate}] (V7) -- (V2);
\draw[g_edge,postaction={decorate}] (V5) -- (V2);
\draw (V1) edge (V2);
\draw (V1) edge (V3);
\draw (V2) edge (V3);
\draw[thick] (V6) edge (V7);
\draw[thick] (V7) edge (V5);

\end{scope}

%\node (com) at (-1,-6) {We have $R_{1}(x)\subseteq R_1(y)$ in any Schnyder wood, independent of wether $(y,x)$ or $(y,z)$ is in $T_3$.};

%\node (com2) at (-2,-7) { If weights were non negative then $x_1\leq y_1$, which is clearly not the case here.};

\end{tikzpicture}
  \caption{\wsnote{A drawing that cannot be realized with positive weights.}}
  \label{fig:eg_neg_weights}
\end{figure}
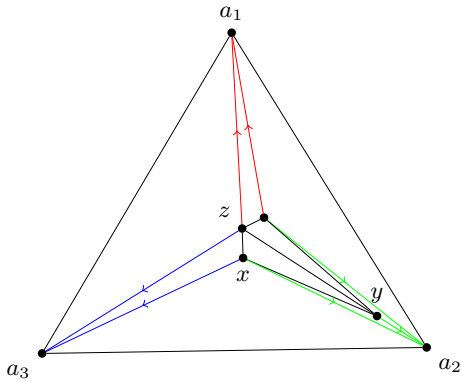

%  Their proof is for the more general setting of Schnyder woods of
%  3-connected planar graphs.)

%In the appendix we give an example of a drawing such that no
%    Schnyder wood results in positive face weights.  

%%%%%%%%%%%%%%%%%%%%%%%%%%%%%%%%%%%%%%%%%
%%%%%%%%%%%%%%%%%%%%%%%%%%%%%%%%%%%%%%%%%
% \section{Identifying weighted Schnyder drawings}
% \label{sec:finding-weights}
 
% It is worth noting that it can be shown that given any planar drawing
% $\Gamma$ of a planar triangulation $T$ and a Schnyder wood $S$ of $T$,
% there exists an assignment of weights $\mathbf{w}$ on $\Fc(T),$ where
% the weights are not necessarily all positive, such that $\Gamma$ is
% precisely the drawing obtained from $S$ and $w$ as described
% in~\ref{eq:weighted_drawing}. The problem of finding such weight
% assignment can be formulated as solving system of linear equations.
% Finally, the existence of solutions for this system can be reduced to
% showing that the set of characteristic vectors $\{\chi_{R_{i}(v)}:
% i\in\{1,2,3\}, v\text{ is an interior vertex of }T\}$ spans a subspace
% of dimension $2n-5$, here $\chi_{R_{i}(v)}$ is a vector whose entries
% are indexed by elements of $\Fc(T)$ and whose entry corresponding to
% $f$ is 1 if and only if $f\in R_{i}(v)$ and 0 otherwise. This is a
% particular case of a result by Felsner and Zickfeld~\cite{Felsner08}.

%%%%%%%%%%%%%%%%%%%%%%%%%%%%%%%%%%%%%%%%%
%%%%%%%%%%%%%%%%%%%%%%%%%%%%%%%%%%%%%%%%%

\section{\pnote{Conclusions and open problems}}
\label{sec:conclusions}

We have made a first step towards morphing straight-line planar graph
drawings with a polynomial number of linear morphs and on a
well-behaved grid.  Our method applies to weighted Schnyder drawings.
\Achange{ There is hope of extending to all straight-line planar
  triangulations.  The first author's thesis~\cite{BarreraCruz14}
  gives partial progress: an algorithm to morph from any straight-line
  planar triangulation to a weighted Schnyder drawing in
  \lFBnote{$4(n-4)$} steps---but not, unfortunately, on a nice grid.}
\lFBnote{This method is simpler than that of Angelini et
  al.~\cite{Angelini14} since the idea is to simply contract vertices
  of degree at most 5 and then uncontract them in reverse order while
  maintaining a Schnyder drawing. No convexifying routines are needed
  since there is no target drawing.}
%\lAnote{[Say more about how this differs from the new $O(n)$ method of Angelini et al.?]}
  
%It might be possible to extend our results to general (non-triangulated) planar graphs 
%using   Felsner's extension~\cite{Felsner01,Felsner03} of Schnyder's results.
\lAnote{
 It is an open question to extend our result to general
  (non-triangulated) planar graphs.   
  This might be possible using the extension of Schnyder's results to 3-connected planar graphs by
 Felsner~\cite{Felsner01,Felsner03}.% might be relevant.
 }

  The problem of efficiently morphing planar graph drawings to
  preserve convexity of faces is wide open---nothing is known besides
  Thomassen's existence result~\cite{Thomassen}.

\bigskip\noindent {\bf Acknowledgements.}  
We thank Stefan Felsner for discussions, David Eppstein for suggestions, and an 
anonymous referee for pointing us to the work of
\Achange{
Bonichon et al.~\cite{Bonichon10}.
% and its relevance to Section~\ref{sec:finding-weights}. 
%Research of 
F.~Barrera-Cruz partially supported by Conacyt. 
%Research of 
P.~Haxell and A.~Lubiw
partially supported by NSERC.
}

%\clearpage
\bibliographystyle{abbrv}
\bibliography{refs}

\newpage

\end{document}